\documentclass[letterpaper, 12pt]{article}[2000/05/19]
\usepackage[english]{babel}
\usepackage{amsfonts,amsmath,amssymb,amsthm,latexsym,amscd,mathrsfs}
\usepackage{ifthen,cite}
\usepackage[bookmarksnumbered=true]{hyperref}
\hypersetup{pdfpagetransition={Split}}

\newcommand{\evenhead}{Author \ name}
\newcommand{\oddhead}{Article \ name}
\newcommand{\theArticleName}{Article \ name}

\newcommand{\FirstPageHeading}[1]{\thispagestyle{empty}%
\noindent\raisebox{0pt}[0pt][0pt]{\makebox[\textwidth]{\protect\footnotesize \sf }}\par}

\newcommand{\ArticleName}[1]{\renewcommand{\theArticleName}{#1}\vspace{-2mm}\par\noindent {\LARGE\bf  #1\par}}
\newcommand{\Author}[1]{\vspace{5mm}\par\noindent {\Large  #1\par} \par\vspace{2mm}\par}
\newcommand{\Address}[1]{\vspace{2mm}\par\noindent {\it #1} \par}
\newcommand{\Email}[1]{\ifthenelse{\equal{#1}{}}{}{\par\noindent {\rm E-mail: }{\it  #1} \par}}
\newcommand{\URLaddress}[1]{\ifthenelse{\equal{#1}{}}{}{\par\noindent {\rm URL: }{\tt  #1} \par}}
\newcommand{\EmailD}[1]{\ifthenelse{\equal{#1}{}}{}{\par\noindent {$\phantom{\dag}$~\rm E-mail: }{\it  #1} \par}}
\newcommand{\URLaddressD}[1]{\ifthenelse{\equal{#1}{}}{}{\par\noindent {$\phantom{\dag}$~\rm URL: }{\tt  #1} \par}}

\newcommand{\Abstract}[1]{\vspace{6mm}\par\noindent\hspace*{10mm}
\parbox{140mm}{\small {\bf Abstract.} #1}\par}
\newcommand{\Keywords}[1]{\vspace{3mm}\par\noindent\hspace*{10mm}
\parbox{140mm}{\small {\bf Key words:} \rm #1}\par}
\newcommand{\Classification}[1]{\vspace{3mm}\par\noindent\hspace*{10mm}
\parbox{140mm}{\small {\it 2000 Mathematics Subject Classification:} \rm #1}\vspace{3mm}\par}
\newcommand{\ShortArticleName}[1]{\renewcommand{\oddhead}{#1}}
\newcommand{\AuthorNameForHeading}[1]{\renewcommand{\evenhead}{#1}}

\long\def\@makecaption#1#2{
  \sbox\@tempboxa{\small \textbf{#1.}\ \ #2}%
  \ifdim \wd\@tempboxa >\hsize
    {\small \textbf{#1.}\ \ #2}\par \else
    \global \@minipagefalse
    \hb@xt@\hsize{\hfil\box\@tempboxa\hfil}%
  \fi \vskip\belowcaptionskip}

\def\numberwithin#1#2{\@ifundefined{c@#1}{\@nocounterr{#1}}{%
  \@ifundefined{c@#2}{\@nocnterr{#2}}{%
  \@addtoreset{#1}{#2}%
  \toks@\@xp\@xp\@xp{\csname the#1\endcsname}%
  \@xp\xdef\csname the#1\endcsname
    {\@xp\@nx\csname the#2\endcsname.\the\toks@}}}}
\def\E^#1{{\buildrel #1 \over\vee}}
\newtheorem{theorem}{Theorem}

\newtheorem{proposition}{Proposition}
{\theoremstyle{definition} 

\newtheorem{remark}{Remark}
}

\begin{document}

\FirstPageHeading{V.I. Gerasimenko}

\ShortArticleName{Derivation of Quantum Kinetic Equations}

\AuthorNameForHeading{V.I. Gerasimenko, Z.A. Tsvir}

\ArticleName{Towards Rigorous Derivation\\ of Quantum Kinetic Equations}

\Author{V.I. Gerasimenko$^\ast$\footnote{E-mail: \emph{gerasym@imath.kiev.ua}}
        and Z.A. Tsvir$^\ast$$^\ast$\footnote{E-mail: \emph{Zhanna.Tsvir@simcorp.com}}}

\Address{$^\ast$\hspace*{2mm}Institute of Mathematics of NAS of Ukraine,\\
    \hspace*{4mm}3, Tereshchenkivs'ka Str.,\\
    \hspace*{4mm}01601, Kyiv-4, Ukraine}

\Address{$^\ast$$^\ast$Taras Shevchenko National University of Kyiv,\\
    \hspace*{4mm}Department of Mechanics and Mathematics,\\
    \hspace*{4mm}2, Academician Glushkov Av.,\\
    \hspace*{4mm}03187, Kyiv, Ukraine}

\bigskip

\Abstract{We develop a rigorous formalism for the description of the evolution of states
of quantum many-particle systems in terms of a one-particle density operator.
For initial states which are specified in terms of a one-particle density operator
the equivalence of the description of the evolution of quantum many-particle states
by the Cauchy problem of the quantum BBGKY hierarchy and by the Cauchy problem of
the generalized quantum kinetic equation together with a sequence of explicitly defined
functionals of a solution of stated kinetic equation is established in the space 
of trace class operators. The links of the specific quantum kinetic equations with 
the generalized quantum kinetic equation are discussed.}

\bigskip

\vspace{2pc}
\Classification{35Q40; 35Q82; 47J35; 82C10; 82C40.}

\makeatletter
\renewcommand{\@evenhead}{
\hspace*{-3pt}\raisebox{-15pt}[\headheight][0pt]{\vbox{\hbox to \textwidth {\thepage \hfil \evenhead}\vskip4pt \hrule}}}
\renewcommand{\@oddhead}{
\hspace*{-3pt}\raisebox{-15pt}[\headheight][0pt]{\vbox{\hbox to \textwidth {\oddhead \hfil \thepage}\vskip4pt\hrule}}}
\renewcommand{\@evenfoot}{}
\renewcommand{\@oddfoot}{}
\makeatother

\newpage
\vphantom{math}

\protect\tableofcontents
\vspace{0.7cm}

\section{Introduction}
It is well known that in certain situations the collective behavior of quantum many-particle systems can be adequately described
by the initial-value problem of the kinetic equation, i.e. by the evolution equation for a one-particle (marginal) density
operator \cite{BQ,CGP97}. To get an understanding of the nature of such phenomenon as the kinetic evolution it is necessary
to answer at least two fundamental questions. One is an origin of initial data for such
evolution equation or in other words what is the immediate cause that many-particle systems
tend to the state described in terms of a one-particle density operator in evolutionary process.
If initial data is completely specified by a one-particle density operator then the other fundamental question is the
derivation of quantum kinetic equations from microscopic dynamics, i.e. from the
von Neumann equation or the quantum BBGKY hierarchy. We note that the main problem herein is whether such intention
can be put on a firm mathematical foundation. In the paper we consider the second question, i.e. the problem of rigorous
derivation of quantum kinetic equations from underlying many-particle dynamics.

First attempt to justify the kinetic equations was undertaken by N.N. Bogolyubov on basis of the perturbation method of construction
of a particular solution of the hierarchy of equations for marginal distribution functions \cite{B46}
(in quantum case in the paper \cite{B47}). Later, drawing an analogy with the equilibrium state expansions,
such approach for classical system of particles was developed in papers of M.S. Green \cite{GM}, M.S. Green and R.A. Piccirelly
\cite{GP,PR} and in series of E.G.D. Cohen papers, summed up in the review \cite{C68} (see also \cite{CE}).
The current view of this problem consists in the following \cite{CIP}.
Since the evolution of states of infinitely many quantum particles is generally described by a sequence of marginal density operators
which is a solution of the initial-value problem of the quantum BBGKY hierarchy, then the evolution of states can be effectively
described by a one-particle density operator governed by the kinetic equation only as a result of some approximations
or in a suitable scaling limit \cite{GH,Sp80,Sp91,Sp07}.
Recently in the framework of such approach the considerable advance in the rigorous derivation
of quantum kinetic equations, namely, the nonlinear Schr\"{o}dinger equation \cite{AA,AGT,BGGM1,ESchY2,FL,M1,ESY10,ESch,GMM,GerUJP}
and the quantum Boltzmann equation \cite{BCEP3,ESY}, is observed.

In the paper we discuss the problem of potentialities inherent in the description of the evolution of
states of many-particle systems in terms of a one-particle density operator.
We demonstrate that in fact if initial data is completely defined by a one-particle marginal density operator, then
all possible states of infinite-particle systems at arbitrary moment
of time can be described within the framework of a one-particle density operator without any approximations.

Now we outline the structure of the paper and the main results.
At first in Section 2 we formulate some definitions and preliminary facts about quantum dynamics of finitely many particles.
Then the main results related to the origin of kinetic evolution is stated.
For initial data specified in terms of trace class operators satisfying a chaos property in case of the
Maxwell-Boltzmann statistics we prove that the Cauchy problem of the quantum BBGKY hierarchy
can be reformulated as a new Cauchy problem for the certain evolution equation for a one-particle marginal density operator
(generalized quantum kinetic equation) and an infinite sequence of explicitly
defined functionals of the solution of this evolution equation which characterizes the correlations of particle states.
In Section 3 we prove the main results, namely, we develop the method of the kinetic cluster expansions of the
cumulants of scattering operators which define the evolution operators of every term of the marginal functional expansions
over the products of a one-particle density operator and derive the generalized quantum kinetic equation.
In Section 4 a solution of the Cauchy problem of the generalized quantum kinetic equation is constructed and
the existence of a strong and a weak solution is proved in the space of trace class operators.
Finally in Section 5 we conclude with some observations and perspectives for future research. Among them we discuss
the problem of the derivation of the specific quantum kinetic equations such as the nonlinear Schr\"{o}dinger equation,
from the constructed generalized quantum kinetic equation in the appropriate scaling limits.
In particular the mean-field scaling limit of a solution of the Cauchy problem
of the generalized quantum kinetic equation and the marginal functionals of the state holds up.

\section{Origin of kinetic evolution}

\subsection{The evolution of many-particle systems: the quantum BBGKY hierarchy}
Hereinafter we consider a quantum system of a non-fixed (i.e. arbitrary but finite) number of the identical (spinless)
particles with unit mass $m=1$ in the space $\mathbb{R}^{\nu},$ $\nu\geq 1$.
The Hamiltonian $H={\bigoplus\limits}_{n=0}^{\infty}H_{n}$ of such system is a self-adjoint operator ($H_{0}=0$)
with the domain $\mathcal{D}(H)=\{\psi=\oplus\psi_{n}\in{\mathcal{F}_{\mathcal{H}}}\mid \psi_{n}\in\mathcal{D}
(H_{n})\in\mathcal{H}^{\otimes n},{\sum\limits}_{n}\|H_{n}\psi_{n}\|^{2}<\infty\}\subset{\mathcal{F}_{\mathcal{H}}},$
where $\mathcal{F}_{\mathcal{H}}={\bigoplus\limits}_{n=0}^{\infty}\mathcal{H}^{\otimes n}$ is the Fock space over
the Hilbert space $\mathcal{H}$. We adopt the usual convention that $\mathcal{H}^{\otimes 0}=\mathbb{C}$.
Assume $\mathcal{H}=L^{2}(\mathbb{R}^\nu)$\,(coordinate representation), then an element $\psi\in\mathcal{F}_{\mathcal{H}}
={\bigoplus\limits}_{n=0}^{\infty}L^{2}(\mathbb{R}^{\nu n})$ is a sequence of functions
$\psi=\big(\psi_0,\psi_{1}(q_1),\ldots,\psi_{n}(q_1,\ldots,q_{n}),\ldots\big)$
such that $\|\psi\|^{2}=|\psi_0|^{2}+\sum_{n=1}^{\infty}\int dq_1\ldots dq_{n}|\psi_{n}(q_1,\ldots,$ $q_{n})|^{2}<+\infty$.
On the subspace of infinitely differentiable functions with compact supports
$\psi_n\in L^{2}_0(\mathbb{R}^{\nu n})\subset L^{2}(\mathbb{R}^{\nu n})$ the
 Hamiltonian $H_{n}$ of $n\geq1$ particles acts according to the formula
\begin{eqnarray}\label{H_Zag}
    &&H_{n}\psi_n = -\frac{\hbar^{2}}{2}\sum\limits_{i=1}^{n}\Delta_{q_i}\psi_n
         +\sum\limits_{i_{1}<i_{2}=1}^{n}\Phi(q_{i_{1}},q_{i_{2}})\psi_{n},
\end{eqnarray}
where $h={2\pi\hbar}$ is a Planck constant, $\Phi$ is a two-body interaction potential satisfying
the Kato conditions \cite{Pe95,Kato}.

The states of finitely many quantum particles belong to the space
$\mathfrak{L}^{1}(\mathcal{F}_\mathcal{H})={\bigoplus\limits}_{n=0}^{\infty}\mathfrak{L}^{1}(\mathcal{H}_{n})$
of the sequences $f=(f_0,f_{1},\ldots,f_{n},\ldots)$ of trace class operators
$f_{n}\equiv f_{n}(1,\ldots,n)\in\mathfrak{L}^{1}(\mathcal{H}_{n})$ and $f_0 \in \mathbb{C}$,
that satisfy the symmetry condition: $f_{n}(1,\ldots,n)=f_{n}(i_{1},\ldots,i_{n})$
for arbitrary $(i_{1},\ldots,i_{n})\in(1,\ldots,n)$, equipped with the norm
\begin{eqnarray*}
    &&\|f\|_{\mathfrak{L}^{1} (\mathcal{F}_\mathcal{H})}=
          \sum\limits_{n=0}^{\infty} \|f_{n}\|_{\mathfrak{L}^{1}(\mathcal{H}_{n})}=
          \sum\limits_{n=0}^{\infty}\mathrm{Tr}_{1,\ldots,n}|f_{n}(1,\ldots,n)|,
\end{eqnarray*}
where $\mathrm{Tr}_{1,\ldots,n}$ are partial traces over $1,\ldots,n$ particles \cite{DauL_5}.
We denote by $\mathfrak{L}^{1}_0(\mathcal{F}_\mathcal{H})=
{\bigoplus\limits}_{n=0}^{\infty}\mathfrak{L}^{1}_0(\mathcal{H}_{n})$ the everywhere dense set
of finite sequences of degenerate operators  with infinitely differentiable kernels
with compact supports \cite{Kato,DauL_5}.

The evolution of states is described by the sequences $F(t)=(F_{1}(t,1),\ldots,F_{s}(t,1,\ldots,s),\ldots)$
of the marginal density operators that satisfy the Cauchy problem of the quantum BBGKY hierarchy
\begin{eqnarray} \label{1}
     &&\frac{d}{dt}F_{s}(t,Y)=-\mathcal{N}_{s}(Y)F_{s}(t,Y) +
           \sum\limits_{i=1}^{s}\,\mathrm{Tr}_{\mathrm{s+1}}
           \big(-\mathcal{N}_{\mathrm{int}}(i,s+1)\big)F_{s+1}(t,Y, s+1),\\ \nonumber\\
     &&F_{s}(t)|_{t=0}=F_{s}^0,\quad s\geq 1,\nonumber
\end{eqnarray}
where $Y\equiv(1,\ldots,s)$, the operator $\mathcal{N}_{s}$
is defined on $\mathfrak{L}^{1}_0(\mathcal{H}_s)$ as follows
\begin{eqnarray}\label{com}
    &&\mathcal{N}_{s}f_s\doteq-\frac{i}{\hbar}\big(f_s\,H_s - H_s\,f_s\big)
\end{eqnarray}
and correspondingly
\begin{eqnarray}\label{cint}
    &&\mathcal{N}_{\mathrm{int}}(i,j)f_s\doteq-\frac{i}{\hbar}\big(f_s\,\Phi(i,j)-\Phi(i,j)\,f_s\big).
\end{eqnarray}

Hereinafter we consider initial data satisfying the factorization property or a "chaos" property \cite{CGP97},
which means the lack of correlations at initial time. For a system of identical particles, obeying the
Maxwell-Boltzmann statistics, we have
\begin{eqnarray}\label{h2}
    &&F(t)|_{t=0}= F^{(c)}\equiv \big(F_1^0(1),\ldots,\prod_{i=1}^s F_1^0(i),\ldots\big).
\end{eqnarray}
The assumption about initial data is intrinsic for the kinetic
description of a gas, because in this case all possible states are
characterized only by a one-particle density operator.

On the space $\mathfrak{L}^{1}(\mathcal{H}_{n})$ we define the group of operators
\begin{eqnarray}\label{groupG}
    &&\mathcal{G}_{n}(-t)f_{n}\,\doteq e^{-{\frac{i}{\hbar}}t H_{n}}\,f_{n}\, e^{{\frac{i}{\hbar}}t H_{n}}.
\end{eqnarray}
On the space $\mathfrak{L}^{1}(\mathcal{H}_{n})$
the mapping: $t\rightarrow\mathcal{G}_{n}(-t)f_n$
is an isometric strongly continuous group which preserves positivity and self-adjointness of operators \cite{BR,GOT}.
For $f_n\in\mathfrak{L}^{1}_{0}(\mathcal{H}_{n})$ there exists a limit in the sense of a strong convergence by which
the infinitesimal generator of the group
of evolution operators (\ref{groupG}) is determined as follows
\begin{eqnarray}\label{infOper}
    &&\lim\limits_{t\rightarrow 0}\frac{1}{t}(\mathcal{G}_n(-t)f_n-f_n)=-\mathcal{N}_n f_n,
\end{eqnarray}
where the operator $(-\mathcal{N}_n)$ is defined by formula (\ref{com}) and
the operator $(-\mathcal{N}_n) f_n$ is defined on the domain $\mathcal{D}(H_n)\subset\mathcal{H}_{n}$.

A solution of the quantum BBGKY hierarchy (\ref{1}) with initial data (\ref{h2}) is represented by the expansion
\begin{eqnarray}\label{srBBGKY}
    &&F_{s}(t,Y)=\sum\limits_{n=0}^{\infty}\frac{1}{n!}\,\mathrm{Tr}_{s+1,\ldots,s+n}\,\,
       \mathfrak{A}_{1+n}(t,\{Y\},s+1,\ldots,s+n)\prod_{i=1}^{s+n} F_1^0(i),
\end{eqnarray}
where the evolution operator $\mathfrak{A}_{1+n}(t)$, $n \geq 0$,
is the $(n+1)$-order cumulant \cite{GOT} of the groups of operators (\ref{groupG})
\begin{eqnarray*}\label{ch}
   &&\mathfrak{A}_{1+n}(t,\{Y\},X\setminus Y )\doteq\sum\limits_{\mathrm{P}:(\{Y\},\,X\setminus Y)
         ={\bigcup\limits}_i X_i}(-1)^{|\mathrm{P}|-1}(|\mathrm{P}|-1)!
         \prod_{X_i\subset \mathrm{P}}\mathcal{G}_{|X_i|}(-t,X_i),
\end{eqnarray*}
and the following notation are used : $\{Y\}$ is the set consisting of one element $Y=(1,\ldots,s)$, i.e. $|\{Y\}|=1$,
$\sum_\mathrm{P}$ is the sum over all possible partitions of the set $(\{Y\},X\setminus Y)\equiv(\{Y\},s+1,\ldots,s+n)$
into $|\mathrm{P}|$ nonempty mutually disjoint subsets $ X_i\subset (\{Y\},X\setminus Y)$.
If $\|F_1^0\|_{\mathfrak{L}^{1}(\mathcal{H}_{1})}<e^{-1}$, series (\ref{srBBGKY}) converges in the norm
of the space $\mathfrak{L}^{1}(\mathcal{H}_{s})$ for arbitrary $t\in\mathbb{R}^1$.

Hereinafter in the capacity of a solution expansion of the quantum BBGKY hierarchy we will use
its equivalent representation in the space $\mathfrak{L}^{1}(\mathcal{F}_\mathcal{H})$,
namely expansion (\ref{srBBGKY}) with the $(n+1)$-order reduced cumulant \cite{Pe95},\cite{GOT}
of the groups of operators (\ref{groupG})
\begin{eqnarray}\label{ch}
   &&\mathfrak{A}_{1+n}(t,\{Y\},s+1,\ldots,s+n)=\sum\limits_{k=0}^{n}(-1)^k\,
       \frac{n!}{k!(n-k)!}\,\mathcal{G}_{s+n-k}(-t).
\end{eqnarray}
If $F_1^0\in\mathfrak{L}^{1}(\mathcal{H}_{1})$, in this case series (\ref{srBBGKY}) converges in the norm of the
space $\mathfrak{L}^{1}(\mathcal{H}_{s})$ for arbitrary $t\in\mathbb{R}^1$ and the estimate holds
\begin{eqnarray}\label{essrBBGKY}
    &&\|F_s(t)\|_{\mathfrak{L}^{1} (\mathcal{H}_{s})}\leq \|F_1^0\|_{\mathfrak{L}^{1} (\mathcal{H}_{1})}^s
       \exp\big({2} \|F_1^0\|_{\mathfrak{L}^{1} (\mathcal{H}_{1})}\big), \quad s\geq 1.
\end{eqnarray}

Thus, in case of initial data (\ref{h2}) the microscopic evolution of states of quantum many-particle
systems is described by sequence (\ref{srBBGKY}),(\ref{ch}).
In the next subsection we formulate the evolution of states in terms of the kinetic theory.

\subsection{The kinetic evolution: main results}
Since we consider initial data (\ref{h2}) which are completely characterized by the one-particle density operator $F_{1}^0$,
namely, $F^{(c)}=\big(F_{1}^0(1),\ldots,\prod_{i=1}^s F_{1}^0(i),\ldots\big)$, the initial-value problem of the quantum
BBGKY hierarchy (\ref{1})-(\ref{h2}) is not completely well-defined Cauchy problem, because the generic initial data is not
independent for every unknown  operator $F_{s}(t,1,\ldots,s),\, s\geq 1$, in the hierarchy of equations.
Thus, it naturally arises the opportunity of reformulating such initial-value problem as a new Cauchy problem for
operator $ F_{1}(t)$ with the independent initial data together with explicitly defined functionals
$F_{s}\big(t,1,\ldots,s \mid F_{1}(t)\big),\,s\geq 2$, of the solution $ F_{1}(t)$ of this Cauchy problem.
We refer to such functionals as the marginal functionals of the state of quantum many-particle systems. At first
we define the restated Cauchy problem.

Functionals $F_{s}\big(t,1,\ldots,s \mid F_{1}(t)\big),\,s\geq 2$, are represented by the following expansions
over products of the one-particle density operator $F_{1}(t)$
\begin{eqnarray}\label{f}
    &&F_{s}\big(t,Y\mid F_{1}(t)\big)\doteq\sum _{n=0}^{\infty }\frac{1}{n!}\,\mathrm{Tr}_{s+1,\ldots,{s+n}}\,
        \mathfrak{V}_{1+n}\big(t,\{Y\},s+1,\ldots,s+n\big)\prod _{i=1}^{s+n} F_{1}(t,i),
\end{eqnarray}
where the $(n+1)$-order evolution operator $\mathfrak{V}_{1+n}(t),\,n\geq0$, are defined as follows
\begin{eqnarray}\label{skrr}
    && \mathfrak{V}_{1+n}(t,\{Y\},X\setminus Y )\doteq n!\,\sum_{k=0}^{n}\,(-1)^k\,\sum_{n_1=1}^{n} \ldots
        \sum_{n_k=1}^{n-n_1-\ldots-n_{k-1}}\frac{1}{(n-n_1-\ldots-n_k)!}\times\\
    &&\times\widehat{\mathfrak{A}}_{1+n-n_1-\ldots-n_k}(t,\{Y\},s+1,\ldots,s+n-n_1-\ldots-n_k)\times\nonumber\\
    &&\times\prod_{j=1}^k\,\sum\limits_{\mbox{\scriptsize $\begin{array}{c}\mathrm{D}_{j}:Z_j=\bigcup_{l_j} X_{l_j},\\
        |\mathrm{D}_{j}|\leq s+n-n_1-\dots-n_j\end{array}$}}\frac{1}{|\mathrm{D}_{j}|!}
        \sum_{i_1\neq\ldots\neq i_{|\mathrm{D}_{j}|}=1}^{s+n-n_1-\ldots-n_j}\,\,
        \prod_{X_{l_j}\subset \mathrm{D}_{j}}\,\frac{1}{|X_{l_j}|!}\,\,
        \widehat{\mathfrak{A}}_{1+|X_{l_j}|}(t,i_{l_j},X_{l_j}),\nonumber
\end{eqnarray}
and $\sum_{\mathrm{D}_{j}:Z_j=\bigcup_{l_j} X_{l_j}}$ is the sum over all possible
dissections\footnote{The dissection $\mathrm{D}$ of the linearly ordered set
$(1,\ldots,n)$ is its partition on connected subsets, $|\mathrm{D}|$ is
the number of subsets of the dissection $\mathrm{D}$. The total number of dissections of an $n$-elements set is $2^{n-1}$.
For example, the set (1,2,3) has four dissections: (1,2,3); ((1),(2,3)); ((1,2),(3)); ((1),(2),(3)).}
$\mathrm{D}_{j}$ of the linearly ordered set
$Z_j\equiv(s+n-n_1-\ldots-n_j+1,\ldots,s+n-n_1-\ldots-n_{j-1})$ on no more than $s+n-n_1-\ldots-n_j$ linearly ordered subsets.
In (\ref{skrr}) we denote by $\widehat{\mathfrak{A}}_{1+n}(t)$ the $(1+n)$-order reduced cumulant, i.e.
\begin{eqnarray*}
   &&\widehat{\mathfrak{A}}_{1+n}(t,\{Y\},s+1,\ldots,s+n)
      =\sum\limits_{k=0}^{n}(-1)^{k}\,\frac{n!}{k! (n-k)!}\,\widehat{\mathcal{G}}_{s+n-k}(t),
\end{eqnarray*}
of the following groups of scattering operators
\begin{eqnarray}\label{so}
   &&\widehat{\mathcal{G}}_{n}(t)=\mathcal{G}_{n}(-t,1,\ldots,n)
      \prod _{i=1}^{n}\mathcal{G}_{1}(t,i), \quad n\geq1.
\end{eqnarray}

We give below for later use a few examples of the evolution operators $\mathfrak{V}_{n}$, $n\geq1$,
of the lower orders:
\begin{eqnarray}\label{rrrls}
   &&\mathfrak{V}_{1}(t,\{Y\})=\widehat{\mathfrak{A}}_{1}(t,\{Y\}),\\
   &&\mathfrak{V}_{2}(t,\{Y\},s+1)=\widehat{\mathfrak{A}}_{2}(t,\{Y\},s+1)-
       \widehat{\mathfrak{A}}_{1}(t,\{Y\})\sum_{i_1=1}^s \widehat{\mathfrak{A}}_{2}(t,i_1,s+1),\nonumber\\
   &&\mathfrak{V}_{3}(t,\{Y\},s+1,s+2)=\widehat{\mathfrak{A}}_{3}(t,\{Y\},s+1,s+2)
       -2!\,\widehat{\mathfrak{A}}_{2}(t,\{Y\},s+1)\times\nonumber\\
   &&\times\sum_{i_1=1}^{s+1}\widehat{\mathfrak{A}}_{2}(t,i_1,s+2)-
       \widehat{\mathfrak{A}}_{1}(t,\{Y\})\big(\sum_{i_1=1}^{s}\widehat{\mathfrak{A}}_{3}(t,i_1,s+1,s+2)-\nonumber\\
   &&-2!\sum_{i_1=1}^{s}\widehat{\mathfrak{A}}_{2}(t,i_1,s+1)\sum_{i_2=1}^{s+1}\widehat{\mathfrak{A}}_{2}(t,i_2,s+2)+
       \sum_{i_1\neq i_2=1}^{s}\widehat{\mathfrak{A}}_{2}(t,i_1,s+1)\widehat{\mathfrak{A}}_{2}(t,i_2,s+2)\big).\nonumber
\end{eqnarray}
In terms of groups of scattering operators (\ref{so}), evolution operators (\ref{rrrls}) are represented:
\begin{eqnarray*}
  &&\mathfrak{V}_{1}(t,\{Y\})=\widehat{\mathcal{G}}_{s}(t,1,\ldots,s),\\
  &&\mathfrak{V}_{2}(t,\{Y\},s+1)=\widehat{\mathcal{G}}_{s+1}(t,1,\ldots,s+1)-\widehat{\mathcal{G}}_{s}(t,1,\ldots,s)
     \sum_{i=1}^{s}\widehat{\mathcal{G}}_{2}(t,i,s+1)+\\
  &&+(s-1)\widehat{\mathcal{G}}_{s}(t,1,\ldots,s).
\end{eqnarray*}
In what follows it will be clear that functionals (\ref{f}) characterize
the correlations of quantum many-particle states.

The one-particle density operator $F_{1}(t)$  is a solution of the following initial-value problem
\begin{eqnarray}
  \label{gke}
    &&\frac{d}{dt}F_{1}(t,1)=-\mathcal{N}_{1}(1)F_{1}(t,1)+\\
    &&+\mathrm{Tr}_{2}\big(-\mathcal{N}_{\mathrm{int}}(1,2)\big)
        \sum\limits_{n=0}^{\infty}\frac{1}{n!}\mathrm{Tr}_{3,\ldots,n+2}
        \mathfrak{V}_{1+n}\big(t,\{1,2\},3,\ldots,n+2\big)\prod _{i=1}^{n+2} F_{1}(t,i),\nonumber\\ \nonumber\\
  \label{2}
    &&F_1(t,1)|_{t=0}= F_1^0(1),
\end{eqnarray}
where the evolution operator $\mathfrak{V}_{1+n}(t)$ is defined by formula (\ref{skrr}). We refer to evolution equation
(\ref{gke}) as the generalized quantum kinetic equation. For systems of classical particles such equation was formulated
in \cite{GP97,GP98,CGP97} and for discrete velocity models in \cite{BGL}.

We observe that the kinetic evolution is described in terms of cumulants of scattering operators (\ref{so})
in contrast to the evolution of states described by the BBGKY hierarchy (\ref{1}).

Thus, the principle of equivalence of initial-value problems (\ref{gke})-(\ref{2}) and (\ref{1}),(\ref{h2}) is true.
\begin{proposition}
In the space $\mathfrak{L}^{1}(\mathcal{F}_\mathcal{H})$ under the condition $\|F_1^0\|_{\mathfrak{L}^{1}(\mathcal{H})}<e^{-2}$
the initial-value problem of the quantum BBGKY hierarchy (\ref{1}),(\ref{h2}) is equivalent to the initial-value problem
of the generalized quantum kinetic equation (\ref{gke}),(\ref{2}) together with the sequence of marginal functionals of the
state $F_{s}\big(t \mid F_{1}(t)\big),\, s\geq 2$, defined by expansions (\ref{f}).
\end{proposition}

The proof of the equivalence proposition is the subject of next section.

It should be noted that the possibility for the corresponding initial data to describe the
evolution of states only within the framework of a one-particle density operator without any approximations
is an inherent property of infinite-particle dynamics.

\begin{remark}
We illustrate the possibility of reformulating of initial-value problem of a hierarchy of evolution equations
in case of depending initial data as a new Cauchy problem for the certain evolution equation
together with explicitly defined functionals of a solution of this Cauchy problem
by the example of the quantum Vlasov hierarchy \cite{Sp80,AA}
\begin{eqnarray}\label{BBGKYlim}
     &&\frac{\partial}{\partial t}f_{s}(t)=\sum\limits_{i=1}^{s}\big(-\mathcal{N}_{1}(i)\big)f_{s}(t)+
         \sum\limits_{i=1}^{s}\mathrm{Tr}_{\mathrm{s+1}}\big(-\mathcal{N}_{\mathrm{int}}(i,s+1)\big)f_{s+1}(t),\\
   \label{BBGKYlim0}
     &&f_{s}(t,1,\ldots,s)|_{t=0}=\prod \limits_{j=1}^{s}f_{1}^0(j),\quad s\geq 1.
\end{eqnarray}
The Cauchy problem (\ref{BBGKYlim})-(\ref{BBGKYlim0}) is equivalent to the Cauchy problem
of the Vlasov quantum kinetic equation
\begin{eqnarray}\label{Vlasov1}
     &&\frac{\partial}{\partial t}f_{1}(t,1)=-\mathcal{N}_{1}(1)f_{1}(t,1)+
       \mathrm{Tr}_{2}\big(-\mathcal{N}_{\mathrm{int}}(1,2)\big)f_{1}(t,1)f_{1}(t,2),\\
   \label{Vlasov2}
     &&f_{1}(t)|_{t=0}=f_{1}^0.
\end{eqnarray}
and a sequence of functionals $f_{s}\big(t,1,\ldots,s \mid f_{1}(t)\big),\, s\geq 2$, defined by the expressions
\begin{eqnarray*}\label{fchaos}
     &&f_{s}\big(t,1,\ldots,s \mid f_{1}(t)\big)=\prod \limits_{j=1}^{s}f_{1}(t,j).
\end{eqnarray*}
The structure of these functionals is usually interpreted as such that the quantum Vlasov hierarchy (\ref{BBGKYlim})
preserves chaos property (\ref{BBGKYlim0}) in time for particles obeying Maxwell-Boltzmann statistics.
\end{remark}

\section{Kinetic evolution of quantum many-particle systems}

\subsection{Marginal functionals of the state}
The straightforward procedure to construct marginal functionals of the state (\ref{f}) consists in
the elimination from expressions of the quantum BBGKY hierarchy solution (\ref{srBBGKY}),(\ref{ch}) for $s=1$ and $s\geq2$
the initial one-particle density operator $F_{1}^0$. With this aim we express the operator $F_{1}^0$ in terms of
the operator $F_{1}(t)$ from expansion (\ref{srBBGKY}),(\ref{ch}) for $s=1$ applying the contraction mapping principle.

In view of expression (\ref{srBBGKY}) for $s=1$ in the space $\mathfrak{L}^{1}(\mathcal{H})$ we have the following
equation for the determination of initial one-particle density operator via the operator $F_1(t)$:
\begin{eqnarray*}\label{eqa}
   &&f=\mathcal{A}(f),
\end{eqnarray*}
where in the space $\mathfrak{L}^{1}(\mathcal{H})$ the nonlinear mapping $\mathcal{A}$ acts according to the formula
\begin{eqnarray}\label{opera}
   &&(\mathcal{A}(f))(1)\doteq f^0-\sum_{n=1}^{\infty}\,\frac{1}{n!}\,
      \mathrm{Tr}_{2,..,n+1}\,\mathcal{G}_1(t,1)\mathfrak{A}_{1+n}(t)\prod_{i=1}^{n+1} f(i),
\end{eqnarray}
and we denote: $f^0\equiv \mathcal{G}_1(t,1)F_1(t,1)$.

Let us find a condition under which nonlinear mapping (\ref{opera}) is a contraction mapping.
Let $f_1$ and $f_2$ are arbitrary elements from the space $\mathfrak{L}^{1}(\mathcal{H})$,
then according to definition (\ref{opera}) we obtain the estimate
\begin{eqnarray*}
   &&\|\mathcal{A}(f_1)-\mathcal{A}(f_2)\|_{\mathfrak{L}^{1}(\mathcal{H})}\leq\sum_{n=1}^{\infty}\frac{2^n}{n!}(n+1)
      (\|f\|_{\mathfrak{L}^{1}(\mathcal{H})})^n\|f_1-f_2\|_{\mathfrak{L}^{1}(\mathcal{H})}=\\
   &&=(e^{2\,\|f\|_{\mathfrak{L}^{1}(\mathcal{H})}}(2\,\|f\|_{\mathfrak{L}^{1}(\mathcal{H})}+1)-1)
      \|f_1-f_2\|_{\mathfrak{L}^{1}(\mathcal{H})},
\end{eqnarray*}
where $\|f\|_{\mathfrak{L}^{1}(\mathcal{H})}=\max(\|f_1\|_{\mathfrak{L}^{1}(\mathcal{H})},\|f_2\|_{\mathfrak{L}^{1}(\mathcal{H})})$.
The mapping $\mathcal{A}$ is contractive under the condition
\begin{eqnarray}\label{conde}
   &&\|f\|_{\mathfrak{L}^{1}(\mathcal{H})}<x_0,
\end{eqnarray}
where $x_0$ is a solution of the equation $e^{2x}(2x+1)=2$ such that $x_0\approx0,18742>e^{-2}$.

Therefore under condition (\ref{conde}) there exists a unique solution of equation (\ref{opera})
in the space $\mathfrak{L}^{1}(\mathcal{H})$.
This solution is determined as the limit of successive approximations $f^{(n)}=\mathcal{A}(f^{(n-1)})$ with the
first approximation $f^{(0)}=f^0\equiv \mathcal{G}_1(t,1)F_1(t,1)$. This solution expresses initial data $F_1^0$
by means of the one-particle density operator $F_1(t)$. Consequently, assembling the evolution operators
before the products of operators $F_1(t)$, we can represent solution (\ref{srBBGKY}) of the quantum BBGKY hierarchy
for $s\geq 2$ as the marginal functionals with respect to the one-particle density operator $F_1(t)$.

Thus, if the norm of initial one-particle density operator $\|F_1^0\|_{\mathfrak{L}^{1}(\mathcal{H})}$
satisfies established condition, i.e.
\begin{eqnarray}\label{cond}
   &&\|F_1^0\|_{\mathfrak{L}^{1}(\mathcal{H})}<e^{-2},
\end{eqnarray}
then in view of estimate (\ref{essrBBGKY}) there exists a sequence of marginal functionals of the state
$F_{s}\big(t\mid F_{1}(t)\big),\, s\geq 2$, which are represented by converged series (\ref{f}).
These functionals satisfy the quantum BBGKY hierarchy (\ref{1}) for $s\geq 2$, if the operator $F_{1}(t)$
is given by expression (\ref{srBBGKY}) for $s=1$.
The condition under which the marginal functionals of the state exist was ascertained in \cite{ZhT}
and in case of classical systems of particles in the paper \cite{GP97}.

\begin{remark}
Within the framework of the kinetic description of evolution of quantum states we have used the space of trace class operators
by reason of the existence of global in time solution of the quantum BBGKY hierarchy \cite{GerSh}. In this space
the condition: $\|F_1^0\|_{\mathfrak{L}^{1}(\mathcal{H})}<e^{-1}$, guarantees the convergence of series (\ref{srBBGKY})
and means that the average number of particles is finite. We can reformulate the convergence condition of series
(\ref{srBBGKY}) as a condition on the parameter characterizing the density $\frac{1}{v}$ of a system (the
average number of particles in a unit volume). In fact if we consider the quantum BBGKY hierarchy (\ref{1}) as the evolution
equation in the thermodynamic limit, then as a result of the renormalization of initial data $F_s^0=\frac{1}{v^s}\widetilde{F}_s^0$,
we obtain expansion (\ref{srBBGKY}) over powers of density $\frac{1}{v}$ which converges under the condition: $\frac{1}{v}<e^{-1}$
\cite{GerSh} or for arbitrary values of $\frac{1}{v}$ in case of reduced cumulants (\ref{ch}).
In this case marginal functionals of the state are represented by converged series (\ref{f}) under the condition: $\frac{1}{v}<e^{-2}$.
We emphasize that intensional spaces for the description of states of infinite-particle systems,
that means the description of kinetic evolution or equilibrium states \cite{Gin}, are different from the exploit space \cite{CGP97}.
\end{remark}

\subsection{Kinetic cluster expansions}
Now we formulate one more method to define the marginal functionals of the state in the explicit form,
namely, we develop the method of kinetic cluster expansions.
The following assertion is valid.
\begin{theorem}
Under condition (\ref{cond}) the marginal density operator $F_{s}(t)$ defined by (\ref{srBBGKY}),(\ref{ch})
for $s\geq2$ and the marginal functional $F_{s}\big(t\mid F_{1}(t)\big)$ defined by (\ref{f}),(\ref{skrr}) are equivalent
if and only if the evolution operators $\mathfrak{V}_{1+n}(t)\, ,n\geq0$, satisfy the following recurrence relations
\begin{eqnarray}\label{rrrl2}
  &&\widehat{\mathfrak{A}}_{1+n}(t,\{Y\},s+1,\ldots,s+n)=\sum_{n_1=0}^{n}\frac{n!}{(n-n_1)!}\,
     \mathfrak{V}_{1+n-n_1}\big(t,\{Y\},s+1,\ldots,\\
  &&s+n-n_1\big)\sum\limits_{\mbox{\scriptsize $\begin{array}{c}\mathrm{D}:Z=\bigcup_l X_l,\\|\mathrm{D}|\leq s+n-n_1\end{array}$}}
     \frac{1}{|\mathrm{D}|!}\,\sum_{i_1\neq\ldots\neq i_{|\mathrm{D}|}=1}^{s+n-n_1}\,\,
     \prod_{X_{l}\subset \mathrm{D}}\,\frac{1}{|X_l|!}\,\,\widehat{\mathfrak{A}}_{1+|X_{l}|}(t,i_l,X_{l})\nonumber,
\end{eqnarray}
where $\sum_{\mathrm{D}:Z=\bigcup_l X_l,\,|\mathrm{D}|\leq s+n-n_1}$ is the sum over all possible dissections
$\mathrm{D}$ of the linearly ordered set $Z\equiv(s+n-n_1+1,\ldots,s+n)$ on no more than $s+n-n_1$ linearly ordered subsets.
\end{theorem}
\begin{proof}
\textit{Necessity.}
To derive recurrence relations (\ref{rrrl2}) we assume that for $s\geq2$ marginal density operators (\ref{srBBGKY}) coincide with
the functionals of a one-particle density operator $F_{s}(t\mid F_{1}(t)),\,s\geq2$. These marginal functionals of the state
are represented in the form of series over particle clusters whose evolution is governed by the corresponding order
evolution operator acting on products of one-particle density operators defined on Hilbert spaces
associated with every particle from the cluster, namely as expansions (\ref{f}).

Observing that in case of $s=1$ for a solution of the quantum BBGKY hierarchy defined by expansion (\ref{srBBGKY})
the following equality holds
\begin{eqnarray}\label{prod2}
  &&\prod _{i=1}^{s+n} F_{1}(t,i)=\sum_{n_1=0}^{\infty}\mathrm{Tr}_{s+n+1,\ldots,s+n+n_1}
     \sum\limits_{\mbox{\scriptsize $\begin{array}{c}\mathrm{D}:Z=\bigcup_k X_k,\\|\mathrm{D}|\leq s+n\end{array}$}}
     \,\,\sum_{i_1<\ldots<i_{|\mathrm{D}|}=1}^{s+n}\,\prod_{X_{k}\subset \mathrm{D}}\,\frac{1}{|X_k|!}\times\\
  &&\times\mathfrak{A}_{1+|X_{k}|}(t,i_k,X_{k})\prod\limits_{\mbox{\scriptsize$\begin{array}{c}{l=1},
     \\l\neq i_1,\ldots, i_{|\mathrm{D}|}\end{array}$}}^{s+n}\mathfrak{A}_1(t,l)\prod_{j=1}^{n+s+n_1}F_1^0(j),\nonumber
\end{eqnarray}
where $\sum_{\mathrm{D}:Z=\bigcup_k X_k,\,|\mathrm{D}|\leq s+n}$ is the sum over all possible dissections $\mathrm{D}$
of the linearly ordered set $Z\equiv(s+n+1,\ldots,s+n+n_1)$ on no more than $s+n$ linearly ordered subsets,
we transform functionals $F_{s}(t\mid F_{1}(t)),\,s\geq2$, to the series over products of initial
one-particle density operators.

Then, equating term by term both series for $F_{s}(t),\,s\geq2$, and for the transformed functionals $F_{s}(t\mid F_{1}(t))$
under the trace sings for the evolution operators acting on the same products of initial data,
we obtain the following recurrence relations for the generating evolution operators of functionals (\ref{f})
in terms of cumulants of groups of operators (\ref{groupG})
\begin{eqnarray}\label{rrrl3}
  &&\mathfrak{A}_{1+n}(t,\{Y\},s+1,\ldots,s+n)=\mathfrak{V}_{1+n}\big(t,\{Y\},s+1,\ldots,s+n\big)+\\
  &&+\sum_{n_1=1}^{n}\frac{n!}{(n-n_1)!}\,\mathfrak{V}_{1+n-n_1}\big(t,\{Y\},s+1,\ldots,s+n-n_1\big)
     \sum\limits_{\mbox{\scriptsize $\begin{array}{c}\mathrm{D}:Z=\bigcup_k X_k,\\|\mathrm{D}|\leq s+n-n_1\end{array}$}}
     \frac{1}{|\mathrm{D}|!}\times\nonumber\\
  &&\times\sum_{i_1\neq\ldots\neq i_{|\mathrm{D}|}=1}^{s+n-n_1}\,
     \prod_{X_{k}\subset \mathrm{D}}\,\frac{1}{|X_k|!}\,\mathfrak{A}_{1+|X_{k}|}(t,i_k,X_{k})
     \prod\limits_{\mbox{\scriptsize$\begin{array}{c}{m=1},
     \\m\neq i_1,\ldots,i_{|\mathrm{D}|}\end{array}$}}^{s+n-n_1}\mathfrak{A}_1(t,m)\nonumber,
\end{eqnarray}
where the linearly ordered set $Z=(s+n-n_1+1,\ldots,s+n)$ is dissected on no more than $s+n-n_1$ linearly ordered subsets.

Under the trace signs recurrence relations (\ref{rrrl3}) are naturally represented in terms of cumulants
of scattering operators as (\ref{rrrl2}). We refer to recurrence relations (\ref{rrrl2})
as the kinetic cluster expansions of reduced cumulants of scattering operators (\ref{so}).

\textit{Sufficiency.}
Using recurrence relations (\ref{rrrl2}), i.e. kinetic cluster expansions of reduced cumulants of scattering operators (\ref{so}),
we construct the expansions of the functionals of a one-particle density operator $F_{s}(t\mid F_{1}(t)),\,s\geq2$,
on basis of solution expansions (\ref{srBBGKY}) of the quantum BBGKY hierarchy.
Indeed, taking into account relations (\ref{rrrl3}), we represent series over the summation index $n$ and
the sum over the summation index $n_1$ as the two-fold series
\begin{eqnarray*}
   &&F_{s}(t,1,\ldots,s)=\sum\limits_{n=0}^{\infty}\frac{1}{n!}\sum_{n_1=0}^{\infty}\mathrm{Tr}_{s+1,\ldots,s+n+n_1}
      \mathfrak{V}_{1+n}\big(t,\{Y\},s+1,\ldots,s+n\big)\hspace*{-2mm}
      \sum\limits_{\mbox{\scriptsize $\begin{array}{c}\mathrm{D}:Z=\bigcup_k X_k,\\|\mathrm{D}|\leq s+n\end{array}$}}
      \hspace*{-2mm}\frac{1}{|\mathrm{D}|!}\times\\
   &&\times\sum_{i_1\neq\ldots\neq i_{|\mathrm{D}|}=1}^{s+n}\,\prod_{X_{k}
      \subset\mathrm{D}}\,\frac{1}{|X_k|!}\,\mathfrak{A}_{1+|X_{k}|}(t,i_k,X_{k})
      \prod\limits_{\mbox{\scriptsize$\begin{array}{c}{l=1},
      \\l\neq i_1,\ldots,i_{|\mathrm{D}|}\end{array}$}}^{s+n}\mathfrak{A}_1(t,l)\prod_{j=1}^{n+s+n_1}F_1^0(j),
\end{eqnarray*}
where $Z\equiv(s+n+1,\ldots,s+n+n_1)$ is the linearly ordered set and it is used the notations introduced above.
The series in the right-hand side converge under condition (\ref{cond}).

In view of formula (\ref{prod2}) we identify the series over the summation index $n_1$ with the
products of one-particle density operators and consequently for $s\geq2$ the following equality holds
\begin{eqnarray*}
   &&F_{s}(t,1,\ldots,s)=\sum\limits_{n=0}^{\infty}\frac{1}{n!}\,\mathrm{Tr}_{s+1,\ldots,s+n}\,\,
       \mathfrak{A}_{1+n}(t,\{Y\},s+1,\ldots,s+n)\prod_{i=1}^{s+n} F_1^0(i)=\\
   &&=\sum\limits_{n=0}^{\infty}\frac{1}{n!}\mathrm{Tr}_{s+1,\ldots, s+n}
       \mathfrak{V}_{1+n}\big(t,\{Y\},s+1,\ldots,s+n\big)
       \prod _{i=1}^{s+n} F_{1}(t,i)=F_{s}(t\mid F_{1}(t)),
\end{eqnarray*}
i.e., if kinetic cluster expansions (\ref{rrrl2}) of cumulants of scattering operators (\ref{so}) hold, then
solution expansions (\ref{srBBGKY}) for $s\geq2$ can be represented
in the form of marginal functionals of the state (\ref{f}).
\end{proof}
We make a few examples of relations (\ref{rrrl2}) of the kinetic cluster expansions:
\begin{eqnarray*}\label{rrrle}
   &&\widehat{\mathfrak{A}}_{1}(t,\{Y\})=\mathfrak{V}_{1}(t,\{Y\}),\\
   &&\widehat{\mathfrak{A}}_{2}(t,\{Y\},s+1)=\mathfrak{V}_{2}(t,\{Y\},s+1)+
       \mathfrak{V}_{1}(t,\{Y\})\sum_{i_1=1}^s \widehat{\mathfrak{A}}_{2}(t,i_1,s+1),\\
   &&\widehat{\mathfrak{A}}_{3}(t,\{Y\},s+1,s+2)=\mathfrak{V}_{3}(t,\{Y\},s+1,s+2)+\\
   &&+2!\,\mathfrak{V}_{2}(t,\{Y\},s+1)\sum_{i_1=1}^{s+1} \widehat{\mathfrak{A}}_{2}(t,i_1,s+2)+\\
   &&+\mathfrak{V}_{1}(t,\{Y\})\big(\sum_{i_1=1}^s \widehat{\mathfrak{A}}_{3}(t,i_1,s+1,s+2)+\sum_{i_1\neq i_2=1}^s
       \widehat{\mathfrak{A}}_{2}(t,i_1,s+1)\widehat{\mathfrak{A}}_{2}(t,i_2,s+2)\big).
\end{eqnarray*}
It is evident that solutions of these relations are given by expressions (\ref{rrrls})
which the evolution operators (\ref{skrr}) of the first, second and third order correspondingly
are determined by in the expansions of marginal functionals of the state  (\ref{f}). In general case solutions
of recurrence relations (\ref{rrrl2}) are given by expressions (\ref{skrr}). This statement is verified
as a result of the substitution of expressions (\ref{skrr}) into recurrence relations (\ref{rrrl2}).

It should be emphasized that in case under consideration, i.e. the absence of correlations at initial time,
the correlations generated by the dynamics of a system are completely governed by evolution operators (\ref{skrr}).

Typical properties for the kinetic description of the evolution of constructed marginal functionals of the state (\ref{f})
are induced by the properties of evolution operators (\ref{skrr}).
Let us indicate some intrinsic properties of the evolution operators $\mathfrak{V}_{1+n}(t),\,n\geq0$,
representative for cumulants (semi-invariants) of group of operators.

Since in case of a system of non-interacting particles for scattering operators (\ref{so})
the equality holds: $\widehat{\mathcal{G}}_{n}(t)=I$,
where $I$ is a unit operator, then we have
\begin{eqnarray*}
    &&\mathfrak{V}_{1+n}(t)=I\delta_{n,0},
\end{eqnarray*}
where $\delta_{n,1}$ is a Kronecker symbol.
Similarly, at initial time $t=0$ it is true: $\mathfrak{V}_{1+n}(0)=I\delta_{n,0}$.

The infinitesimal generator of the first-order evolution operator (\ref{rrrls}) is defined by the following
limit in the sense of the norm convergence in the space $\mathfrak{L}^{1}(\mathcal{H}_{n})$
\begin{eqnarray*}
    &&\lim\limits_{t\rightarrow 0}\frac{1}{t}(\mathfrak{V}_{1}(t,\{1,\ldots,n\})-I)f_{n}
       =\sum_{i<j=1}^n(-\mathcal{N}_{\mathrm{int}}(i,j))f_{n},
\end{eqnarray*}
where the operator $(-\mathcal{N}_{\mathrm{int}}(i,j))$ is defined by formula (\ref{cint}) for $f_n\in\mathfrak{L}^{1}_0(\mathcal{H}_{n})\subset\mathfrak{L}^{1}(\mathcal{H}_{n})$.

In general case, i.e. $n\geq2$, in the sense of the norm convergence in the space $\mathfrak{L}^{1}(\mathcal{H}_{n})$
for the $n$-order evolution operator (\ref{skrr}) it holds
\begin{eqnarray*}
    &&\lim\limits_{t\rightarrow 0}\frac{1}{t}\mathfrak{V}_{n}(t,1,\ldots,n) f_{n}=0.
\end{eqnarray*}

Summarize we observe that in case of initial data (\ref{h2}) which is completely characterized
by the one-particle density operator $F_{1}^0$, solution (\ref{srBBGKY}) for $s\geq2$ of the quantum BBGKY hierarchy (\ref{1})
and marginal functionals of the state (\ref{f}) give two equivalent approaches to
the description of states of quantum many-particle systems.

\subsection{The derivation of the generalized quantum kinetic equation}
Let us construct the evolution equation which satisfies expression (\ref{srBBGKY}),(\ref{ch}) for $s=1$ .

Taking into account equality (\ref{infOper}) and observing
the validity of the following equalities for reduced cumulants (\ref{ch}) of groups (\ref{groupG}) for $f\in\mathfrak{L}^{1}_{0}(\mathcal{F}_\mathcal{H})$ in the sense of the norm convergence
(for $n\geq2$ it is a consequence that we consider a system of particles interacting by a two-body potential):
\begin{eqnarray*}
   &&\lim\limits_{t\rightarrow 0}\frac{1}{t}\,\mathrm{Tr}_{2}\,\mathfrak{A}_{2}(t,1,2)f_{2}(1,2)
       =\mathrm{Tr}_{2}\big(-\mathcal{N}_{\mathrm{int}}(1,2)\big)f_{2}(1,2),\\ \\
   &&\lim\limits_{t\rightarrow 0}\frac{1}{t}\,\mathrm{Tr}_{2,\ldots,n+1}\,\mathfrak{A}_{1+n}(t,1,\ldots,n+1)f_{n+1}=0,\quad n\geq2,
\end{eqnarray*}
we will differentiate over the time variable expression (\ref{srBBGKY}),(\ref{ch}) for $s=1$
in the sense of pointwise convergence in the space $\mathfrak{L}^{1}(\mathcal{H}_{1})$. As result it holds
\begin{eqnarray}\label{de}
  &&\frac{d}{dt}F_{1}(t,1)=-\mathcal{N}_1(1))F_{1}(t,1)+\\
  &&+\mathrm{Tr}_{2}(-\mathcal{N}_{\mathrm{int}}(1,2))\sum\limits_{n=0}^{\infty}\frac{1}{n!}\mathrm{Tr}_{3,\ldots,n+2}
      \,\mathfrak{A}_{1+n}(t,\{1,2\},3,\ldots,n+2)\prod_{i=1}^{n+2} F_1^0(i).\nonumber
\end{eqnarray}
In second summand in the right-hand side of this equality we expand reduced cumulants (\ref{ch})
of groups (\ref{groupG}) into transformed (\ref{rrrl3}) kinetic cluster
expansions (\ref{rrrl2}) and represent series over the summation index $n$
and the sum over the summation index $n_1$ as the two-fold series. Then the following equalities take place:
\begin{eqnarray*}
  &&\sum\limits_{n=0}^{\infty}\frac{1}{n!}\,\mathrm{Tr}_{2,\ldots,n+2}(-\mathcal{N}_{\mathrm{int}}(1,2))
      \mathfrak{A}_{1+n}(t,\{1,2\},3,\ldots,n+2)\prod_{i=1}^{n+2} F_1^0(i)=\\
  &&=\sum\limits_{n=0}^{\infty}\frac{1}{n!}\,\mathrm{Tr}_{2,\ldots,n+2}(-\mathcal{N}_{\mathrm{int}}(1,2))
      \sum_{n_1=0}^{n}\frac{n!}{(n-n_1)!}\,\mathfrak{V}_{1+n-n_1}\big(t,\{1,2\},3,\ldots,n+2-n_1\big)\times\nonumber\\
  &&\times\sum_{\mathrm{D}:Z=\bigcup_l X_l}\frac{1}{|\mathrm{D}|!}\,\sum_{i_1\neq\ldots\neq i_{|\mathrm{D}|}=1}^{n+2-n_1}\,
      \,\prod_{X_{l}\subset\mathrm{D}}\frac{1}{|X_l|!}\,\mathfrak{A}_{1+|X_{l}|}(t,i_l,X_{l})
      \prod\limits_{\mbox{\scriptsize$\begin{array}{c}{m=1},\\m\neq i_1,\ldots,i_{|\mathrm{D}|}\end{array}$}}^{2+n-n_1}
      \mathfrak{A}_1(t,m)\prod_{i=1}^{n+2} F_1^0(i)=\\
  &&=\mathrm{Tr}_{2}(-\mathcal{N}_{\mathrm{int}}(1,2))\sum\limits_{n=0}^{\infty}
      \frac{1}{n!}\,\mathrm{Tr}_{3,\ldots,n+2}\mathfrak{V}_{1+n}\big(t,\{1,2\},3,\ldots,n+2\big)
      \sum_{n_1=0}^{\infty}\,\sum_{\mathrm{D}:Z^{'}=\bigcup_l X_l}\frac{1}{|\mathrm{D}|!}\times\\
  &&\times\sum_{i_1\neq\ldots\neq i_{|\mathrm{D}|}=1}^{n+2}
      \,\prod_{X_{l}\subset\mathrm{D}}\frac{1}{|X_l|!}\,\mathfrak{A}_{1+|X_{l}|}(t,i_l,X_{l})
      \prod\limits_{\mbox{\scriptsize$\begin{array}{c}{m=1},\\m\neq i_1,\ldots,i_{|\mathrm{D}|}\end{array}$}}^{n+2}
      \mathfrak{A}_1(t,m)\prod_{i=1}^{n+2+n_1}F_1^0(i),
\end{eqnarray*}
where $Z\equiv(n+3-n_1,\ldots,n+2)$ and $Z^{'}\equiv(n+3,\ldots,n+2+n_1)$ are linearly ordered sets
and it is used the notations accepted above.

Consequently, applying in case of $s=2$ formula (\ref{prod2}) to the obtained expression,
from equality (\ref{de}) we derive
\begin{eqnarray}\label{eq}
  &&\frac{d}{dt}F_{1}(t,1)=-\mathcal{N}_1(1)F_{1}(t,1)+\\
  &&+\mathrm{Tr}_{2}(-\mathcal{N}_{\mathrm{int}}(1,2))\sum _{n=0}^{\infty}\frac{1}{n!}\,\mathrm{Tr}_{3,\ldots,{n+2}}
     \,\mathfrak{V}_{1+n}\big(t,\{1,2\},3,\ldots,n+2\big)\prod _{i=1}^{n+2} F_{1}(t,i).\nonumber
\end{eqnarray}
Under condition (\ref{cond}) the series in right-hand side of this equality converges.

The constructed identity (\ref{eq}) for the one-particle (marginal) density operator $F_{1}(t,1)$
we will treat as the evolution equation which governs the one-particle states
of many-particle quantum systems obeying the Maxwell-Boltzmann statistics.

We remark that one more approach to the derivation of the generalized quantum kinetic equation consists in
its construction on the basis of dynamics of correlations \cite{GerSh, GP10}.

Thus, if initial data is completely defined by a one-particle density operator, then
all possible states of infinite-particle systems at arbitrary moment of time can be described within
the framework of a one-particle density operator without any approximations. In other words, for mentioned states
the evolution of states governed by the quantum BBGKY hierarchy (\ref{1}) can be completely described by
the generalized quantum kinetic equation (\ref{gke}) and therefore Proposition 1 is valid.

\subsection{Some properties of marginal functionals of the state}
We indicate that expansions (\ref{f}) of marginal functionals of the state are nonequilibrium analog
of the Mayer-Ursell expansions over powers of the density of equilibrium
marginal density operators \cite{B46,PR}.

In case of the description of states in terms of the marginal correlation operators \cite{GerSh,GP10}
\begin{eqnarray*}
    &&G_{s}(t,Y)=\sum_{\mbox{\scriptsize $\begin{array}{c}\mathrm{P}:Y=\bigcup_{i}X_{i}\end{array}$}}
        (-1)^{|\mathrm{P}|-1}(|\mathrm{P}|-1)!\,\prod_{X_i\subset \mathrm{P}}F_{|X_i|}(t,X_i),
\end{eqnarray*}
where $\sum_\mathrm{P}$ is the sum over all possible partitions $\mathrm{P}$ of the set $Y\equiv(1,\ldots,s),\,s\geq2$,
into $|\mathrm{P}|$ nonempty mutually disjoint subsets $ X_i\subset Y$ and in particular, $G_{1}(t)=F_{1}(t)$,
the marginal correlation functionals $G_{s}\big(t,Y\mid F_{1}(t)\big),\,s\geq2$, are represented by the expansions
similar to (\ref{f}), namely
\begin{eqnarray}\label{cf}
    &&G_{s}\big(t,Y\mid F_{1}(t)\big)=\sum\limits_{n=0}^{\infty}\frac{1}{n!}\,\mathrm{Tr}_{s+1,\ldots, s+n}
        \mathfrak{V}_{1+n}\big(t,\theta(\{Y\}),s+1,\ldots,s+n\big)\prod _{i=1}^{s+n} F_{1}(t,i).
\end{eqnarray}
In expansion (\ref{cf}) it is introduced the notion of the declasterization mapping $\theta: \{Y\}\rightarrow Y$.
This mapping is defined by the formula \cite{P10}
\begin{eqnarray*}
    &&\theta(\{Y\})=Y,
\end{eqnarray*}
that it means the declasterization of particle clusters in cumulants of scattering operators, i.e.
in contrast to expansion (\ref{f}) the $n$ term of expansions (\ref{cf}) of marginal correlation functionals
$G_{s}\big(t,1,\ldots,s\mid F_{1}(t)\big)$ is governed by the $(1+n)$-order evolution operator (\ref{skrr})
of the $(s+n)$-order, $n\geq0$, cumulants of the scattering operators, for example, as compared to (\ref{rrrls})
the lower orders evolution operators $\mathfrak{V}_{1+n}\big(t,\theta(\{Y\}),s+1,\ldots,s+n\big), \,n\geq0$, have the form
\begin{eqnarray}\label{rrrlsc}
   &&\mathfrak{V}_{1}(t,\theta(\{Y\}))=\widehat{\mathfrak{A}}_{s}(t,\theta(\{Y\}),\\
   &&\mathfrak{V}_{2}(t,\theta(\{Y\}),s+1)=\widehat{\mathfrak{A}}_{s+1}(t,\theta(\{Y\}),s+1)-
       \widehat{\mathfrak{A}}_{s}(t,\theta(\{Y\}))\sum_{i=1}^s \widehat{\mathfrak{A}}_{2}(t,i,s+1),\nonumber
\end{eqnarray}
and in case of $s=2$, it holds
\begin{eqnarray*}
    &&\mathfrak{V}_{1}(t,\theta(\{1,2\}))=\widehat{\mathcal{G}}_{2}(t,1,2)-I.
\end{eqnarray*}

In the framework of the description of states by marginal functionals of the state (\ref{f}) the average values, for example,
of the additive-type observables $A^{(1)}=(0,a_{1},\ldots,{\sum\limits}_{i=1}^{n}a_{1}(i),\ldots)$ are given by the functional
\begin{eqnarray}\label{averageg}
    &&\langle A^{(1)}\rangle(t)=\mathrm{Tr}_{1}\,a_{1}(1)F_{1}(t,1),
\end{eqnarray}
i.e. they are defined by a solution of the generalized quantum kinetic equation (\ref{gke}),
or in general case of the $s$-particle observables $A^{(s)}=(0,\ldots,0,a_{s}(1,\ldots,s),\ldots,
\sum_{i_{1}<\ldots<i_{s}=1}^{n}a_s(i_{1},\ldots,i_{s}),\ldots)$ by the functional
\begin{eqnarray*}
    &&\langle A^{(s)}\rangle(t)=\frac{1}{s!}\,\mathrm{Tr}_{1,\ldots,s}\,a_{s}(1,\ldots,s)
       F_{s}\big(t,1,\ldots,s\mid F_{1}(t)\big), \quad s\geq2.
\end{eqnarray*}
For $A^{(s)}\in \mathfrak{L}(\mathcal{F}_\mathcal{H})$ and $F_1(t)\in \mathfrak{L}^{1}(\mathcal{H})$
these functionals exist.

The dispersion of an additive-type observable is defined by a solution of the generalized quantum kinetic
equation (\ref{gke}) and marginal correlation functionals (\ref{cf}) as follows
\begin{eqnarray*}
    &&\langle(A^{(1)}-\langle A^{(1)}\rangle(t))^2\rangle(t)=\\
    &&=\mathrm{Tr}_{1}\,(a_1^2(1)-\langle A^{(1)}\rangle^2(t))F_{1}(t,1)+
       \mathrm{Tr}_{1,2}\,a_{1}(1)a_{1}(2)G_{2}\big(t,1,2\mid F_{1}(t)\big),
\end{eqnarray*}
where the functional $\langle A^{(1)}\rangle(t)$ is determined by expression (\ref{averageg}). Note that
the dispersion of observables is minimal for states characterized by marginal correlation functionals (\ref{cf})
equals to zero, i.e. from macroscopic point of view the evolution of many-particle states
with the minimal dispersion is the Markovian kinetic evolution.

In fact functionals (\ref{cf}) or (\ref{f}) characterize the correlations of states
of quantum many-particle systems. We illustrate close links of functionals (\ref{cf})
and (\ref{f}) in the following way:
\begin{eqnarray*}
    &&F_{2}\big(t,1,2\mid F_{1}(t)\big)=F_{1}(t,1)F_{1}(t,2)+G_{2}\big(t,1,2\mid F_{1}(t)\big).
\end{eqnarray*}
Basically this equality gives the classification of all possible currently in use
scaling limits \cite{GH,Sp80,Sp91}. In the scaling limits it is assumed that chaos property (\ref{h2})
of initial state preserves in time, i.e. the scaling limit means such limit of dimensionless parameters
of a system in which the marginal correlation functional $G_{2}\big(t,1,2\mid F_{1}(t)\big)$ vanishes.
According to definition (\ref{cf}), it is possible, if particles of every finite particle cluster move
without collisions (\ref{rrrlsc}). In conclusions the mean-field scaling limit of functionals (\ref{f})
and (\ref{cf}) holds up.

Another approach to the derivation of the Markovian kinetic equations was formulated by Bogolyubov \cite{B46} (see also \cite{CU})
and consists in the construction of marginal functionals of the state $F_{s}\big(t,Y\mid F_{1}(t)\big)$ by the perturbation method.

Let us consider first two terms of expansion (\ref{f}).
If an interaction potential in (\ref{H_Zag}) is a bounded operator and $f_{s+1}\in\mathfrak{L}^{1}(\mathcal{H}_{s+1})$, then
for the second-order cumulant $\widehat{\mathfrak{A}}_{2}(t,\{Y\},s+1)$ of scattering operators (\ref{so})
an analog of the Duhamel equation holds
\begin{eqnarray}\label{dcum}
    &&\widehat{\mathfrak{A}}_{2}(t,\{Y\},s+1)f_{s+1}=\int_0^t d\tau\,\mathcal{G}_{s}(-\tau,Y)
       \mathcal{G}_{1}(-\tau,s+1)\sum\limits_{i_1=1}^{s}\big(-\mathcal{N}_{\mathrm{int}}(i_1,s+1)\big)\times\\
    &&\times\widehat{\mathcal{G}}_{s+1}(\tau-t,Y,s+1)\prod_{i_2=1}^{s+1}\mathcal{G}_{1}(\tau,i_2)f_{s+1},\nonumber
\end{eqnarray}
and, consequently, for the second-order evolution operator $\mathfrak{V}_{2}(t,\{Y\},s+1)$ we have
\begin{eqnarray}\label{Df}
    &&\mathfrak{V}_{2}(t,\{Y\},s+1)f_{s+1}\doteq\big(\widehat{\mathfrak{A}}_{2}(t,\{Y\},s+1)-
       \widehat{\mathfrak{A}}_{1}(t,\{Y\})\sum_{i_1=1}^s \widehat{\mathfrak{A}}_{2}(t,i_1,s+1)\big)f_{s+1}=\\
    &&=\int_0^t d\tau\,\mathcal{G}_{s}(-\tau,Y)\mathcal{G}_{1}(-\tau,s+1)\big(\sum\limits_{i_1=1}^{s}
       (-\mathcal{N}_{\mathrm{int}}(i_1,s+1))\widehat{\mathcal{G}}_{s+1}(\tau-t,Y,s+1)-\nonumber\\
    &&-\widehat{\mathcal{G}}_{s}(\tau-t,Y)\sum\limits_{i_1=1}^{s}(-\mathcal{N}_{\mathrm{int}}(i_1,s+1))
       \widehat{\mathcal{G}}_{2}(\tau-t,i_1,s+1)\big)\prod_{i_2=1}^{s+1}\mathcal{G}_{1}(\tau,i_2)f_{s+1}.\nonumber
\end{eqnarray}
In the kinetic (macroscopic) scale of the variation of variables \cite{CIP} groups of operators (\ref{groupG})
of finitely many particles depend on microscopic time variable $\varepsilon^{-1}t$, where $\varepsilon\geq0$ is a scale
parameter, and the dimensionless marginal functionals of the state are represented in the form:
$F_{s}\big(\varepsilon^{-1}t,Y\mid F_{1}(t)\big)$.
Note that on the macroscopic scale the typical length for the kinetic phenomena described, for example,
by the quantum Boltzmann equation is the mean free pass. Then according to (\ref{Df})
in the formal Markovian limit $\varepsilon\rightarrow0$ the first two terms of the dimensionless
marginal functional expansions coincide with corresponding terms constructed by the perturbation
method with the use of the weakening of correlation condition in \cite{B46} (see also \cite{B47,CU})
\begin{eqnarray}\label{mark}
    &&\lim\limits_{\epsilon\rightarrow 0}F_{s}\big(\varepsilon^{-1}t,Y\mid F_{1}(t)\big)=
       \widehat{\mathcal{G}}_{s}(\infty,Y)\prod _{i=1}^{s} F_{1}(t,i)+\\
    &&+\int_0^{\infty} d\tau\,\mathcal{G}_{s}(-\tau,Y)\mathrm{Tr}_{s+1}\big(\sum\limits_{i_1=1}^{s}
       (-\mathcal{N}_{\mathrm{int}}(i_1,s+1))\widehat{\mathcal{G}}_{s+1}(\infty,Y,s+1)-\nonumber\\
    &&-\widehat{\mathcal{G}}_{s}(\infty,Y)\sum\limits_{i_1=1}^{s}(-\mathcal{N}_{\mathrm{int}}(i_1,s+1))
       \widehat{\mathcal{G}}_{2}(\infty,i_1,s+1)\big)\prod_{i_2=1}^{s+1}\mathcal{G}_{1}(\tau,i_2)F_{1}(t,i_2)+etc.\nonumber
\end{eqnarray}

Therefore in the kinetic scale the collision integral of the generalized kinetic equation (\ref{gke}) takes the form of
Bogolyubov's collision integral \cite{B46,CU} which enables to control correlations of infinite-particle systems.
We remark that in the homogeneous case the collision integral of the first approximation in (\ref{mark}) has
a more general form than the quantum Boltzmann collision integral.

\section{Initial-value problem of generalized kinetic equation}

\subsection{An existence theorem}
Before considering abstract initial-value problem (\ref{gke})-(\ref{2}) in the space $\mathfrak{L}^{1}(\mathcal{H})$
we generalize it for case of $n$-body interaction potential $\Phi^{(n)}, n\geq1$. In this case the Cauchy problem
of the generalized quantum kinetic equation has the form
\begin{eqnarray}\label{gkeN}
    &&\frac{d}{dt}F_{1}(t,1)=-\mathcal{N}_{1}(1)F_{1}(t,1)+
       \sum\limits_{n=1}^{\infty}\sum _{k=1}^{n}\frac{1}{(n-k)!}\frac{1}{k!}
       \,\mathrm{Tr}_{2,\ldots,n+1}(-\mathcal{N}_{\mathrm{int}}^{(k+1)})(1,\nonumber\\
    &&\ldots,k+1)\mathfrak{V}_{1+n-k}(t,\{1,\ldots,k+1\},k+2,\ldots,n+1)\prod _{i=1}^{n+1} F_{1}(t,i),\\
\label{2N}
    &&F_1(t,1)|_{t=0}= F_1^0(1),
\end{eqnarray}
where $\mathfrak{V}_{1+n-k}(t)$, is the $(1+n-k)$-order evolution operator (\ref{skrr})
and notations (\ref{com}),(\ref{cint}) are used, and
\begin{eqnarray*}
    &&\mathcal{N}_{\mathrm{int}}^{(n)}f_n\doteq-\frac{i}{\hbar}\big(f_n\,\Phi^{(n)}-\Phi^{(n)}\,f_n\big).
\end{eqnarray*}
The collision integral in the generalized quantum kinetic equation (\ref{gkeN}) is defined
by the convergent series under condition (\ref{cond}).

For the sake of a comparison of the structure of various collision integral components in (\ref{gkeN}) we give
expressions of the collision integral term describing a two-body interaction and three particle correlations
\begin{eqnarray*}
    &&\mathrm{Tr}_{2,3}(-\mathcal{N}_{\mathrm{int}}^{(2)})(1,2)
       \mathfrak{V}_{2}(t,\{1,2\},3)F_{1}(t,1)F_{1}(t,2)F_{1}(t,3),
\end{eqnarray*}
and the collision integral term describing a three-body interaction
\begin{eqnarray*}
    &&\frac{1}{2!}\mathrm{Tr}_{2,3}(-\mathcal{N}_{\mathrm{int}}^{(3)})(1,2,3)
       \mathfrak{V}_{1}(t,\{1,2,3\})F_{1}(t,1)F_{1}(t,2)F_{1}(t,3),
\end{eqnarray*}
where the evolution operators $\mathfrak{V}_{2}(t,\{1,2\},3)$ and $\mathfrak{V}_{1}(t,\{1,2,3\})$
are defined by formulas (\ref{rrrls}).

For the Cauchy problem (\ref{gkeN})-(\ref{2N}) ( and (\ref{gke})-(\ref{2})) in the space $\mathfrak{L}^{1}(\mathcal{H})$
the following statement is true.
\begin{theorem}
The global in time solution of initial-value problem (\ref{gkeN})-(\ref{2N})
is determined by the following expansion
\begin{eqnarray}\label{ske}
    &&F_{1}(t,1)= \sum\limits_{n=0}^{\infty}\frac{1}{n!}\,\mathrm{Tr}_{2,\ldots,{1+n}}\,\,
        \mathfrak{A}_{1+n}(t,1,\ldots,n+1)\prod _{i=1}^{n+1}F_{1}^0(i),
\end{eqnarray}
where the reduced cumulants $\mathfrak{A}_{1+n}(t),\, n\geq0,$ are defined by formula (\ref{ch}).
If $\|F_1^0\|_{\mathfrak{L}^{1}(\mathcal{H})}<e^{-2}$, then for $F_1^0\in\mathfrak{L}^{1}_{0}(\mathcal{H})$ it is
a strong (classical) solution and for an arbitrary initial data $F_1^{0}\in\mathfrak{L}^{1}(\mathcal{H})$
it is a weak (generalized) solution.
\end{theorem}
\begin{proof}
Let $F_1^0\in\mathfrak{L}^{1}_{0}(\mathcal{H})$. It will be recalled that series (\ref{ske}) converges
in the norm of the space $\mathfrak{L}^{1}(\mathcal{H})$ and estimate (\ref{essrBBGKY}) holds.
Series (\ref{ske}) is a strong solution of initial-value problem (\ref{gkeN})-(\ref{2N}),
if the equality holds
\begin{eqnarray}\label{sts}
    &&\lim_{\triangle t\rightarrow 0}\big\|\frac{1}{\triangle t}\,\big(F_{1}(t+\triangle t,1)-F_{1}(t,1)\big)-\\
    &&-\big(-\mathcal{N}_{1}(1)F_{1}(t,1)+\sum\limits_{n=1}^{\infty}\sum _{k=1}^{n}\frac{1}{(n-k)!}\frac{1}{k!}
       \,\mathrm{Tr}_{2,\ldots,n+1}(-\mathcal{N}_{\mathrm{int}}^{(k+1)})(1,\ldots,k+1)\times\nonumber\\
    &&\times\mathfrak{V}_{1+n-k}(t,\{1,\ldots,k+1\},k+2,\ldots,n+1)\prod _{i=1}^{n+1}
       F_{1}(t,i)\big)\big\|_{\mathfrak{L}^{1}(\mathcal{H})}=0,\nonumber
\end{eqnarray}
where abridged notations are applied: the symbols $F_{1}(t,1)$ and $\prod _{i=1}^{n+1}F_{1}(t,i)$
are implied series (\ref{ske}) and for $s=1$ series (\ref{prod2}), respectively.

To prove the existence of a strong solution of initial-value problem (\ref{gkeN})-(\ref{2N}) we use the result of section 3.4
on the differentiation of expansion (\ref{ske}) over time variable in the sense of the pointwise convergence
in the space $\mathfrak{L}^{1}(\mathcal{H})$ with a little modification.
Taking into account that for $n\geq1$ and $f_{n+1}\in\mathfrak{L}^{1}_{0}(\mathcal{H}_{n+1})$ the equality is true
\begin{eqnarray*}
   &&\lim\limits_{t\rightarrow 0}\big\|\frac{1}{t}\,\mathrm{Tr}_{2,\ldots,n+1}\,\mathfrak{A}_{1+n}(t,1,\ldots,n+1)f_{n+1}-
      \mathrm{Tr}_{2,\ldots,n+1}\,(-\mathcal{N}_{\mathrm{int}}^{(n+1)})(1,\ldots,n)f_{n+1}\big\|_{\mathfrak{L}^{1}(\mathcal{H})}=0,
\end{eqnarray*}
in the sense of the pointwise convergence in the space $\mathfrak{L}^{1}(\mathcal{H})$ we have
\begin{eqnarray}\label{den}
    &&\lim_{\triangle t\rightarrow 0}\frac{1}{\triangle t}\,\big(F_{1}(t+\triangle t,1)
      -F_{1}(t,1)\big)=-\mathcal{N}_{1}F_{1}(t)+\\
    &&+\sum\limits_{n=1}^{\infty}\frac{1}{n!} \,\mathrm{Tr}_{2,\ldots,n+1}(-\mathcal{N}_{\mathrm{int}}^{(n+1)})(1,\ldots,n+1)
      \sum\limits_{k=0}^{\infty}\frac{1}{k!}\,\mathrm{Tr}_{n+2,\ldots,{n+k+1}}\,\,\mathfrak{A}_{1+k}(t,\nonumber\\
    &&\{1,\dots,n+1\},n+2,\ldots,n+k+1)\prod _{i=1}^{n+k+1}F_{1}^0(i).\nonumber
\end{eqnarray}
In the second summand in the right-hand side of this equality we expand the reduced cumulants (\ref{ch})
of groups (\ref{groupG}) into transformed (\ref{rrrl3}) kinetic cluster expansions (\ref{rrrl2})
and represent series over the summation index $n$ and the sum over the summation index $k$ as the two-fold series.
Then, applying formula (\ref{prod2}) in case of $s=n+1$ to the obtained expression, from equality (\ref{den}) we derive
\begin{eqnarray}\label{denn}
    &&\sum\limits_{n=1}^{\infty}\frac{1}{n!} \,\mathrm{Tr}_{2,\ldots,n+1}(-\mathcal{N}_{\mathrm{int}}^{(n+1)})(1,\ldots,n+1)
      \sum\limits_{k=0}^{\infty}\frac{1}{k!}\,\mathrm{Tr}_{n+2,\ldots,{n+k+1}}\,\mathfrak{A}_{1+k}(t,\{1,\dots,n+1\},\\
    &&n+2\ldots,n+k+1)\prod _{i=1}^{n+k+1}F_{1}^0(i)=\sum\limits_{n=1}^{\infty}\sum _{k=1}^{n}
      \frac{1}{(n-k)!}\frac{1}{k!}\mathrm{Tr}_{2,\ldots,n+1}(-\mathcal{N}_{\mathrm{int}}^{(k+1)})(1,\nonumber\\
    &&\ldots,k+1)\,\mathfrak{V}_{1+n-k}(t,\{1,\ldots,k+1\},k+2,\ldots,n+1)\prod _{i=1}^{n+1}F_{1}(t,i).\nonumber
\end{eqnarray}
Under the condition: $\|F_1^0\|_{\mathfrak{L}^{1}(\mathcal{H})}<e^{-2}$, the series in the right-hand side
of this equality converges in the norm of the space $\mathfrak{L}^{1}(\mathcal{H})$.
Hence in view of equalities (\ref{den}) and (\ref{denn})
for $F_1^0\in\mathfrak{L}^{1}_{0}(\mathcal{H})$, we finally establish the validity of equality (\ref{sts}).

The proof that for arbitrary $F_1^0\in\mathfrak{L}^{1}(\mathcal{H})$
a weak solution of initial-value problem (\ref{gkeN})-(\ref{2N}) is given by expansion (\ref{ske})
is the subject of next section.
\end{proof}

\subsection{A weak solution}
Let us prove that in case of arbitrary initial data $F_1^{0}\in\mathfrak{L}^{1}(\mathcal{H})$
expansion (\ref{ske}) is a weak solution of the initial-value problem of the generalized quantum kinetic equation (\ref{gkeN}).
With this purpose we introduce the functional
\begin{eqnarray}\label{funcc}
    &&(f_{1},F_{1}(t))\doteq \mathrm{Tr}_{1}\,f_{1}(1)\,F_{1}(t,1),
\end{eqnarray}
where $f_{1}\in \mathfrak{L}_{0}(\mathcal{H})\subset\mathfrak{L}(\mathcal{H})$
is degenerate bounded operator with infinitely times differentiable
kernel with compact support and the operator $F_{1}(t)$ is defined by expansion (\ref{ske}).
According to estimate (\ref{essrBBGKY}), for $f_{1}\in \mathfrak{L}_{0}(\mathcal{H})$, functional (\ref{funcc})
exists and represents by the convergence series.

Using expansion (\ref{ske}), we transform  functional (\ref{funcc}) as follows
\begin{eqnarray*}\label{funk-gN}
    &&(f_{1},F_{1}(t))=\sum\limits_{n=0}^{\infty}\frac{1}{n!}\,\mathrm{Tr}_{1,\ldots,{1+n}}\,f_{1}(1)\,
        \mathfrak{A}_{1+n}(t,1,\ldots,n+1)\prod _{i=1}^{n+1} F_1^0(i)=\\ \nonumber
    &&=\sum\limits_{n=0}^{\infty}\frac{1}{n!}\,\mathrm{Tr}_{1,\ldots,{1+n}}\,\sum\limits_{k=0}^{n}(-1)^k\,
        \frac{n!}{k!(n-k)!}\,\mathcal{G}_{1+n-k}(t)f_{1}(1)\,\prod _{i=1}^{n+1} F_1^0(i),
\end{eqnarray*}
where the group of operators $\mathcal{G}_{1+n-k}(t)$ is adjoint to the group $\mathcal{G}_{1+n-k}(-t)$
in the sense of functional (\ref{funcc}). For $F_1^{0}\in\mathfrak{L}^{1}(\mathcal{H})$ and
$f_{1}\in \mathfrak{L}_{0}(\mathcal{H})$ considering (\ref{den}) the following equality holds
in the sense of the $\ast$-weak convergence \cite{BR} of the space $\mathfrak{L}(\mathcal{H})$
\begin{eqnarray}\label{d_funk-d}
    &&\lim\limits_{\triangle t\rightarrow0}\sum\limits_{n=0}^{\infty}\frac{1}{n!}\,
        \mathrm{Tr}_{1,\ldots,{1+n}}\,\sum\limits_{k=0}^{n}(-1)^k\,
        \frac{n!}{k!(n-k)!}\,\frac{1}{\triangle t}(\mathcal{G}_{1+n-k}(t+\triangle t)-\\
    &&-\mathcal{G}_{1+n-k}(t))f_{1}(1)\,\prod _{i=1}^{n+1} F_1^0(i)=\nonumber\\
    &&=(\mathcal{N}_{1}f_{1},F_{1}(t))+\sum\limits_{n=1}^{\infty}\frac{1}{n!}
       \,\mathrm{Tr}_{1,\ldots,n+1}\,\mathcal{N}_{\mathrm{int}}^{(n+1)}(1,\ldots,n+1)f_{1}(1)\,\times\nonumber\\
    &&\times\sum\limits_{k=0}^{\infty}\frac{1}{k!}\,\mathrm{Tr}_{n+2,\ldots,{n+k+1}}\,\,
       \mathfrak{A}_{1+k}(t,\{1,\dots,n+1\},n+2,\ldots,n+k+1)\prod _{i=1}^{n+k+1}F_{1}^0(i).\nonumber
\end{eqnarray}
For $F_1^{0}\in\mathfrak{L}^{1}(\mathcal{H})$ and bounded interaction potentials the limit functionals exist.
Using equality (\ref{denn}), we transform the second functional in (\ref{d_funk-d}) to the form
\begin{eqnarray}\label{d_funk-dn}
    &&\sum\limits_{n=1}^{\infty}\frac{1}{n!}\,\mathrm{Tr}_{1,\ldots,n+1}\,
       \mathcal{N}_{\mathrm{int}}^{(n+1)}f_{1}(1)\,\sum\limits_{k=0}^{\infty}\frac{1}{k!}\,
       \mathrm{Tr}_{n+2,\ldots,{n+k+1}}\,\,\mathfrak{A}_{1+k}(t)\prod _{i=1}^{n+k+1}F_{1}^0(i)=\\
   &&=\sum\limits_{n=1}^{\infty}\sum _{k=1}^{n}
       \frac{1}{(n-k)!}\frac{1}{k!}\,\mathrm{Tr}_{1,\ldots,n+1}\,\mathcal{N}_{\mathrm{int}}^{(k+1)}(1,\ldots,k+1)
       f_{1}(1)\,\mathfrak{V}_{1+n-k}(t)\prod _{i=1}^{n+1}F_{1}(t,i).\nonumber
\end{eqnarray}

Therefore as consequence of equalities (\ref{d_funk-d}),(\ref{d_funk-dn}), for functional (\ref{funcc}) we have
\begin{eqnarray}\label{d_funk-gNr}
    &&\frac{d}{dt}(f_{1},F_{1}(t))=(\mathcal{N}_{1}f_{1},F_{1}(t))+\\
    &&+\sum\limits_{n=1}^{\infty}\,\mathrm{Tr}_{1,\ldots,n+1}\,\sum _{k=1}^{n}
       \frac{1}{(n-k)!}\frac{1}{k!}\,\mathcal{N}_{\mathrm{int}}^{(k+1)}
       f_{1}(1)\,\mathfrak{V}_{1+n-k}(t)\prod _{i=1}^{n+1}F_{1}(t,i).\nonumber
\end{eqnarray}
Equality (\ref{d_funk-gNr}) means that expansion (\ref{ske}) for arbitrary $F_1^{0}\in\mathfrak{L}^{1}(\mathcal{H})$
is a weak solution of the Cauchy problem (\ref{gkeN})-(\ref{2N}).

For the Cauchy problem (\ref{gkeN})-(\ref{2N}) it can be introduced the notion of a weak solution in certain generalized sense.
Consider the functional
\begin{eqnarray}\label{func-g}
  &&\big(f,F(t\mid F_{1}(t))\big)\doteq \sum_{s=0}^{\infty}\,\frac{1}{s!}
      \,\mathrm{Tr}_{\mathrm{1,\ldots,s}}\,f_{s}\,F_{s}(t\mid F_{1}(t)),
\end{eqnarray}
where $f=(f_0,f_1,\ldots,f_n,\ldots)\in \mathfrak{L}_{0}(\mathcal{F}_\mathcal{H})\in\mathfrak{L}(\mathcal{F}_\mathcal{H})$
is a finite sequence of degenerate bounded operators \cite{Kato}
with infinitely times differentiable kernels with compact supports and elements of the sequence
$F\big(t,\mid F_{1}(t)\big)\doteq \big(F_{1}(t,1),F_{2}(t,1,2\mid F_{1}(t)),
\ldots,F_{s}\big(t,1,\ldots,s \mid F_{1}(t)\big),\ldots\big)$
are defined by formulas (\ref{ske}) and (\ref{f}) for the first and other elements correspondingly.
If for functional (\ref{func-g}) it is valid the equality
\begin{eqnarray}\label{w}
    &&\frac{d}{dt}\big(f,F(t\mid F_{1}(t))\big)=\big({\mathcal{B}}^{+}f,F(t\mid F_{1}(t))\big),
\end{eqnarray}
where ${\mathcal{B}}^{+}$ is the operator dual to the generator of the quantum BBGKY hierarchy \cite{BG}, i.e.
\begin{eqnarray*}\label{wg}
    &&({\mathcal{B}}^{+}f)_{s}(Y)\doteq \mathcal{N}_{s}(Y)f_{s}(Y)+\\
    &&+\sum\limits_{n=1}^{s}\frac{1}{n!}\sum\limits_{k=n+1}^s \frac{1}{(k-n)!}\sum_{j_1\neq\ldots\neq j_{k}=1}^s
        \mathcal{N}_{\mathrm{int}}^{(k)}(j_1,\ldots,j_{k})f_{s-n}(Y\backslash(j_1,\ldots,j_{n})),
\end{eqnarray*}
we are said to be that expansion (\ref{ske}) is a weak solution
of the Cauchy problem (\ref{gkeN})-(\ref{2N}) in extended meaning.

To verify this definition we transform functional (\ref{func-g}) as follows \cite{BG}
\begin{eqnarray*}
  &&\big(f,F(t\mid F_{1}(t))\big)= \sum_{s=0}^{\infty}\,\frac{1}{s!}
      \,\mathrm{Tr}_{\mathrm{1,\ldots,s}}\,\sum_{n=0}^s\,\frac{1}{(s-n)!}\sum_{j_1\neq\ldots\neq j_{s-n}=1}^s\,
      \sum\limits_{\substack{Z\subset Y\backslash (j_1,\ldots,j_{s-n})}}\,
      (-1)^{|Y\backslash (j_1,\ldots,j_{s-n})\backslash Z|}\\
  &&\times\mathcal{G}_{s-n+|Z|}(t,(j_1,\ldots,j_{s-n})\cup Z )\,
      f_{s-n}(j_1,\ldots,j_{s-n})\,\prod _{i=1}^{s}F_{1}^0(i),
\end{eqnarray*}
where ${\sum\limits}_{\substack{Z\subset Y\backslash (j_1,\ldots,j_{s-n})}}$
is a sum over all subsets $Z\subset Y\backslash (j_1,\ldots,j_{s-n})$ of the set $
Y\backslash (j_1,\ldots,j_{s-n})\subset(1,\ldots,s)$. For $F_1^{0}\in\mathfrak{L}^{1}(\mathcal{H})$
and bounded interaction potentials this functional exists.

Skipping the details, as a result for $f\in \mathfrak{L}_{0}(\mathcal{F}_\mathcal{H})$
the derivative of functional (\ref{func-g})
over the time variable in the sense of the $\ast$-weak convergence in the space
$\mathfrak{L}(\mathcal{F}_\mathcal{H})$ takes the form \cite{BG}
\begin{eqnarray}\label{d_funk-gN}
    &&\frac{d}{dt}\big(f,F(t\mid F_{1}(t))\big)=\sum\limits_{s=0}^{\infty}\frac{1}{s!}\,
       \mathrm{Tr}_{1,\ldots,{s}}\big(\mathcal{N}_{s}(Y)f_{s}(Y)+\sum\limits_{n=1}^{s}\frac{1}{n!}
       \sum\limits_{k=n+1}^s \frac{1}{(k-n)!}\times\\
    &&\times\sum_{j_1\neq\ldots\neq j_{k}=1}^s\mathcal{N}_{\mathrm{int}}^{(k)}(j_1,\ldots,j_{k})
       f_{s-n}(Y\backslash (j_1,\ldots,j_{n}))\big)\,F_{s}(t,Y\mid F_{1}(t)).\nonumber
\end{eqnarray}
In the sense of defined notion of a weak solution in extended meaning (\ref{w}) equality (\ref{d_funk-gN})
means that for arbitrary initial data $F_1^{0}\in\mathfrak{L}^{1}(\mathcal{H})$ a weak solution
of the initial-value problem of the generalized quantum kinetic equation (\ref{gkeN})
is determined by formula (\ref{ske}).

\section{Conclusions}
We demonstrate that in fact if initial data is completely defined by a one-particle density operator, then
all possible states of infinite-particle systems at arbitrary moment of time can be described within the framework
of a one-particle density operator without any approximations and explicitly defined functionals of this one-particle
density operator. One of the advantage of such approach is the possibility to construct the kinetic equations
in scaling limits if there are correlations of particle states at initial time \cite{CGP97}, for instance,
correlations characterizing the condensate states \cite{BQ}.

The specific quantum kinetic equations such as the Boltzmann equation and other ones, can be derived
from the constructed generalized quantum kinetic equation in the appropriate scaling limits or as a result of certain
approximations. For example, in the mean-field limit \cite{Sp91} (the case of scaled interaction potential
$\epsilon\,\Phi$, i.e. $\epsilon\,\mathcal{N}_{\mathrm{int}}$) we derive the quantum Vlasov equation
and for pure states the Hartree equation or the nonlinear Schr\"{o}dinger equation
(in case of a two-body interaction potential with the cubic
nonlinear term and for $n$-body interaction potential (\ref{gkeN}) with the $2n-1$ power nonlinear term).

Indeed, if there exists the following limit $f_{1}^0\in\mathfrak{L}^{1}(\mathcal{H}_{1})$ of initial data (\ref{2})
\begin{eqnarray}\label{0lim}
     &&\lim\limits_{\epsilon\rightarrow 0}\big\|\epsilon\,F_{1}^0-f_{1}^0
         \big\|_{\mathfrak{L}^{1}(\mathcal{H}_{1})}=0,
\end{eqnarray}
then for arbitrary finite time interval, there exists the limit
of solution (\ref{ske}) of the generalized quantum kinetic equation (\ref{gke})
\begin{eqnarray}\label{1lim}
     &&\lim\limits_{\epsilon\rightarrow 0} \big\|\epsilon \,F_{1}(t)-
         f_{1}(t)\big\|_{\mathfrak{L}^{1}(\mathcal{H}_{1})}=0,
\end{eqnarray}
where $f_{1}(t)$ is a strong solution of the Cauchy problem of the quantum Vlasov equation (\ref{Vlasov1})-(\ref{Vlasov2})
represented in the form of the following expansion
\begin{eqnarray}\label{viter}
     &&f_{1}(t,1)=\sum\limits_{n=0}^{\infty}\,\int\limits_0^tdt_{1}\ldots\int\limits_0^{t_{n-1}}dt_{n}\,
         \mathrm{Tr}_{\mathrm{s+1,\ldots,s+n}}\,\prod\limits_{j=1}^{s}\mathcal{G}_{1}(-t+t_{1},j)\times\\
     &&\times \sum\limits_{i_{1}=1}^{s}(-\mathcal{N}_{\mathrm{int}}(i_{1},s+1))
         \prod\limits_{j_1=1}^{s+1}\mathcal{G}_{1}(-t_{1}+t_{2},j_1)\ldots
         \prod\limits_{j_{n-1}=1}^{s+n-1}\mathcal{G}_{1}(-t_{n-1}+t_{n},j_{n-1})\times\nonumber\\
     &&\times\sum\limits_{i_{n}=1}^{s+n-1}(-\mathcal{N}_{\mathrm{int}}(i_{n},s+n))
         \prod\limits_{j_n=1}^{s+n}\mathcal{G}_{1}(-t_{n},j_n)\prod\limits_{i=1}^{s+n}f_{1}^0(i),\nonumber
\end{eqnarray}
and the operator $\mathcal{N}_{\mathrm{int}}$ is defined by formula (\ref{cint}).
For bounded interaction potentials series (\ref{viter}) converges for finite time interval \cite{Pe95}.

This statement is a consequence that, if $f_{s}\in\mathfrak{L}^{1}(\mathcal{H}_{s})$, then for arbitrary finite
time interval for the strongly continuous group (\ref{groupG}) it holds \cite{Kato}
\begin{eqnarray*}\label{lemma1}
    &&\lim\limits_{\epsilon\rightarrow 0}\big\|\mathcal{G}_{s}(-t)f_{s}-
         \prod\limits_{j=1}^{s}\mathcal{G}_{1}(-t,j)f_{s}\big\|_{\mathfrak{L}^{1}(\mathcal{H}_{s})}=0,
\end{eqnarray*}
and in general case the validity of the following equality
\begin{eqnarray*}\label{Duam2}
    &&\lim\limits_{\epsilon\rightarrow 0}\big\|\frac{1}{\epsilon^n}\,
        \mathfrak{A}_{1+n}(t,\{1,\ldots,s\},s+1,\ldots,s+n)f_{s+n}-\nonumber\\
    &&-\int\limits_0^tdt_{1}\ldots\int\limits_0^{t_{n-1}}dt_{n} \prod\limits_{j=1}^{s}\mathcal{G}_{1}(-t+t_{1},j)
        \sum\limits_{i_{1}=1}^{s}(-\mathcal{N}_{\mathrm{int}}(i_{1},s+1))
        \prod\limits_{j_1=1}^{s+1}\mathcal{G}_{1}(-t_{1}+t_{2},j_1)\ldots\nonumber\\
    &&\ldots\prod\limits_{j_{n-1}=1}^{s+n-1}\mathcal{G}_{1}(-t_{n-1}+t_{n},j_{n-1})
        \sum\limits_{i_{n}=1}^{s+n-1}(-\mathcal{N}_{\mathrm{int}}(i_{n},s+n))
        \prod\limits_{j_n=1}^{s+n}\mathcal{G}_{1}(-t_{n},j_n)f_{s+n}\big\|_{\mathfrak{L}^{1}(\mathcal{H}_{s+n})}=0.
\end{eqnarray*}

Then according to definition (\ref{skrr}) of the evolution operators
$\mathfrak{V}_{1+n}(t,\{1,\ldots,s\},s+1,\ldots,s+n),\, n\geq0,$ from expansion (\ref{f}),
taking into account an analog of the Duhamel equation for scattering operators
\begin{eqnarray*}
    &&\big(\widehat{\mathcal{G}}_{s}(t,1,\ldots,s)-I\big)f_s=\epsilon\int\limits_{0}^{t}d\tau
         \prod\limits_{l=1}^{s}\mathcal{G}_{1}(\tau,l)\big(-\sum\limits_{i<j=1}^{s}
         \mathcal{N}_{\mathrm{int}}(i,j)\big)\mathcal{G}_{s}(-\tau)f_s,
\end{eqnarray*}
and (\ref{dcum}), we establish formulas of an asymptotic perturbation of evolution operators (\ref{skrr})
\begin{eqnarray*}
    &&\lim\limits_{\epsilon\rightarrow 0}\big\|\big(\mathfrak{V}_{1}(t,\{1,\ldots,s\})-I\big)f_{s}
        \big\|_{\mathfrak{L}^{1}(\mathcal{H}_{s})}=0,
\end{eqnarray*}
and for $n\geq1$, correspondingly
\begin{eqnarray*}
    &&\lim\limits_{\epsilon\rightarrow 0}\big\|\frac{1}{\epsilon^n}\,
        \mathfrak{V}_{1+n}(t,\{1,\ldots,s\},s+1,\ldots,s+n)f_{s+n}\big\|_{\mathfrak{L}^{1}(\mathcal{H}_{s+n})}=0.
\end{eqnarray*}

Since a solution of initial-value problem (\ref{gke})-(\ref{2}) of the generalized quantum kinetic equation converges
to a solution of initial-value problem (\ref{Vlasov1})-(\ref{Vlasov2}) of the Vlasov quantum kinetic equation
as (\ref{0lim}),(\ref{1lim}), for functional (\ref{f}) for every $s\geq2$ it is true
\begin{eqnarray*}
    &&\lim\limits_{\epsilon\rightarrow 0} \big\|\epsilon^{s} F_{s}\big(t,1,\ldots,s \mid F_{1}(t)\big)-
        \prod\limits_{j=1}^{s}f_{1}(t,j)\big\|_{\mathfrak{L}^{1}(\mathcal{H}_{s})}=0,
\end{eqnarray*}
where $f_{1}(t)$ is defined by series (\ref{viter}) which converges for finite time interval, or
for marginal correlation functionals (\ref{cf}) it holds
\begin{eqnarray*}
    &&\lim\limits_{\epsilon\rightarrow 0} \big\|\epsilon^{s}
        G_{s}\big(t,1,\ldots,s\mid F_{1}(t)\big)\big\|_{\mathfrak{L}^{1}(\mathcal{H}_{s})}=0.
\end{eqnarray*}
The last equalities mean that in the mean-field scaling limit chaos property (\ref{h2}) preserves in time.

In the case of quantum systems of particles obeying Fermi or Bose statistics \cite{GP10}
the generalized quantum kinetic equation (\ref{gkeN}) and functionals (\ref{f}) have different structures.
The analysis of these cases will be given in a separate paper.

In the end it should be emphasized that a one-particle marginal density operator which belongs to the space
$\mathfrak{L}^{1}_{\alpha}(\mathcal{F}_\mathcal{H})$ describes only finitely many particles, i.e. systems
for which the average number of particles in a system is finite.
In order to describe the evolution of infinitely many particles we have to construct solutions for initial data
that belongs to more general Banach spaces than the space of trace class operators \cite{CGP97}.
For example, it can be the space of sequences of bounded operators containing the equilibrium states \cite{Pe95,Gin}.
In that case every term of the solution expansions of the quantum BBGKY hierarchy (\ref{1}) and correspondingly
of the generalized quantum kinetic equation (\ref{gke}) as well as
marginal functionals of the state (\ref{f}) contains the divergent traces \cite{CGP97,C68}.

\addcontentsline{toc}{section}{References}
\renewcommand{\refname}{References}

\end{document}